\documentclass[11pt]{article}
    \usepackage[margin = 1.0 in]{geometry}
    \date{}

\usepackage{paper}
\usepackage{xcolor}
\newcommand{\new}[1]{{\color{teal} #1}}
\renewcommand{\new}{}

\title{Moderate Dimension Reduction for $k$-Center Clustering\footnote{A preliminary version appeared in SoCG 2024 \cite{JiangKS24}.
  \new{
    This full version includes omitted proofs and additional material, 
    particularly \Cref{sec:streaming_low_dim,sec:doubling,sec:fpq_appendix,sec:all_pairs_moderate_dim},
    the proofs of \Cref{lem:ball_expansion_gaussians} and \Cref{thm:constrained_k_center}, and a discussion on the optimality of \Cref{thm:dim_reduction_informal}.}
  }}

    \author{Shaofeng H.-C. Jiang\thanks{
        Research partially supported by a national key R\&D program of China No. 2021YFA1000900 and
        a startup fund from Peking University.
        Email: \texttt{shaofeng.jiang@pku.edu.cn}
      }
      \\ Peking University
      \and Robert Krauthgamer\thanks{Work partially supported 
          by the Israel Science Foundation grant \#1336/23
          and the Weizmann Data Science Research Center.
          Email: \texttt{robert.krauthgamer@weizmann.ac.il} 
        } 
        \\ Weizmann Institute of Science 
        \and Shay Sapir\thanks{Partially supported by the Israeli Council for Higher Education (CHE) via the Weizmann Data Science Research Center.
            Email: \texttt{shay.sapir@weizmann.ac.il} 
        } 
        \\ Weizmann Institute of Science
    }

\begin{document}

\maketitle

\begin{abstract}
The Johnson-Lindenstrauss (JL) Lemma introduced the concept of dimension reduction via a random linear map, which has become a fundamental technique in many computational settings. For a set of $n$ points in $\mathbb{R}^d$ and any fixed $\epsilon>0$, it reduces the dimension $d$ to $O(\log n)$ while preserving, with high probability, all the pairwise Euclidean distances within factor $1+\epsilon$. Perhaps surprisingly, the target dimension can be lower if one only wishes to preserve the optimal value of a certain problem on the pointset, e.g., Euclidean max-cut or $k$-means. However, for some notorious problems, like diameter (aka furthest pair), dimension reduction via the JL map to below $O(\log n)$ does not preserve the optimal value within factor $1+\epsilon$.

We propose to focus on another regime, of \emph{moderate dimension reduction}, where a problem's value is preserved within factor $\alpha>1$ using target dimension $\log n / \mathrm{poly}(\alpha)$. We establish the viability of this approach and show that the famous $k$-center problem is $\alpha$-approximated when reducing to dimension $O(\tfrac{\log n}{\alpha^2}+\log k)$. Along the way, we address the diameter problem via the special case $k=1$. Our result extends to several important variants of $k$-center (with outliers, capacities, or fairness constraints), and the bound improves further with the input's doubling dimension. 

While our $\mathrm{poly}(\alpha)$-factor improvement in the dimension may seem small, it actually has significant implications for streaming algorithms, and easily yields an algorithm for $k$-center in dynamic geometric streams, that achieves $O(\alpha)$-approximation using space $\mathrm{poly}(kdn^{1/\alpha^2})$. This is the first algorithm to beat $O(n)$ space in high dimension $d$, as all previous algorithms require space at least $\exp(d)$. Furthermore, it extends to the $k$-center variants mentioned above.
\end{abstract}

\section{Introduction}

The seminal work of Johnson and Lindenstrauss~\cite{article_JohnsonLindenstrauss}
introduced the technique of dimension reduction via an (oblivious) random linear map,
and this technique has become fundamental in many computational settings, 
from offline to streaming and distributed algorithms,
especially nowadays that high-dimensional data is ubiquitous.  
Their so-called JL Lemma asserts (roughly) that for any fixed $\epsilon>0$,
a random mapping (e.g., projection) of a set $P\subset\R^d$ of $n$ points to target dimension $t=O(\log n)$
preserves the Euclidean distances between all points in $P$ within $1+\epsilon$ factor.
Furthermore, the JL Lemma is known to be tight~\cite{DBLP:conf/focs/LarsenN17},
see also the recent survey~\cite{Nelson20}.

Perhaps surprisingly,
the target dimension can sometimes be reduced below that $O(\log n)$ bound, 
particularly when one only wants to preserve the optimal value of a specific objective function
rather than all pairwise distances. 
Indeed, for several optimization problems on the pointset $P$, 
previous work has shown that the target dimension
may be much smaller than $O(\log n)$ or even independent of $n$, e.g., for Euclidean max-cut \cite{DBLP:conf/wads/LammersenSS09, lammersen2010approximation,DBLP:conf/stoc/ChenJK23}, $k$-means \cite{DBLP:conf/nips/BoutsidisZD10, DBLP:conf/stoc/CohenEMMP15, DBLP:conf/stoc/BecchettiBC0S19,DBLP:conf/stoc/MakarychevMR19}, Wasserstein barycenter~\cite{IzzoSZ21} and subspace approximation \cite{DBLP:journals/corr/CharikarWaingarten23}.
However, for some notorious problems, 
like facility location, minimum spanning tree~\cite{DBLP:conf/icml/NarayananSIZ21}
and diameter (aka furthest pair), 
dimension reduction via the JL map to below $O(\log n)$
does not preserve the optimal value within factor $1+\epsilon$.

\medskip

In light of this, we consider a new regime of 
\emph{moderate dimension reduction},
where a problem's value is approximated within factor $O(1)$
using target dimension slightly below $O(\log n)$.\footnote{A related approach called low-quality dimension reduction
was proposed in~\cite{DBLP:journals/talg/Anagnostopoulos18}.
Similarly to ours, it refers to a map with target dimension below $O(\log n)$, but its
utility is to $(1+\epsilon)$-approximate Nearest-Neighbor Search (NNS),
and since NNS is not an optimization problem, 
our definitions and techniques are rather different. 
}
More precisely, we aim at the following tradeoff:
achieving $\alpha$-estimation,\footnote{
We say that dimension reduction map $f$ achieves $\alpha$-estimation if $1\leq \opt(f(P))\leq \alpha\opt(P)$, where $\opt$ is the optimal value of the problem;
and
we say it achieves $\alpha$-approximation
if it also preserves solutions, in a sense that we define formally later.
Our results extend to $O(\alpha)$-approximation, but we focus here on $O(\alpha)$-estimation for clarity.}
for any desired $\alpha>1$,
when reducing to dimension $\log n/\poly(\alpha)$.
This relaxation of the approximation factor, from $1+\epsilon$ to $\alpha$,
may be effective in combating the ``curse of dimensionality'' phenomenon,
because when an algorithm's efficiency is exponential in the dimension
(e.g., space complexity in streaming algorithms), 
bounds of the form $2^{O(\log n)} = n^{O(1)}$
improve to $2^{\log n/\poly(\alpha)}=n^{1/\poly(\alpha)}$.
This tradeoff can even yield target dimension $t=O(1)$,
by using approximation $\alpha = \polylog(n)$,
and this can lead to new results (e.g., streaming algorithms).
\new{Note that the interesting regime for the target dimension bound
is indeed $\log n/\poly(\alpha)$. 
A larger target dimension $t=O(\tfrac{\log n}{\log \alpha})$, 
can preserve \emph{all the pairwise distances} within factor $\alpha$, 
see \Cref{sec:all_pairs_moderate_dim} for a proof, 
however this bound might be less useful, e.g., to get target dimension $t=O(1)$, the approximation becomes $\alpha=\poly(n)$.}

We study this regime of moderate dimension reduction
for the fundamental problem of $k$-center clustering:
The input is a set $P\subset \R^d$ of $n$ points,
and the goal is to find a set of centers $C\subset \R^d$ of size $k$ that minimizes the objective
$\max_{p\in P}\dist(p,C)$,
where we denote $\dist(p,C) = \min_{c\in C}\|p-c\|$
(throughout, we use $\ell_2$ norm). The special case $k=1$ is exactly
the minimum enclosing ball problem,
which is within factor $2$ of the diameter problem,
and for both problems, we show that reducing the dimension to $o(\log n)$
is unlikely to achieve $(1+\epsilon)$-estimation. 
Our main contribution is
a general framework for moderate dimension reduction
that works even 
for more challenging and widely studied variants of $k$-center,
such as the outliers~\cite{DBLP:conf/soda/CharikarKMN01}, 
capacitated~\cite{Bar-IlanKP93,DBLP:journals/siamdm/KhullerS00},
and fair~\cite{DBLP:conf/nips/Chierichetti0LV17} variants.

\subsection{Main Results}
\label{sec:main_results}

Our main result is a moderate dimension reduction for $k$-center
via an (oblivious) random linear map. 
For simplicity, 
we consider a map 
defined via
a matrix $G\in\R^{d\times t}$ of iid Gaussians (scaled appropriately),
although our result may possibly extend to other JL maps,
similarly to prior work in this context~\cite{DBLP:conf/stoc/MakarychevMR19,DBLP:conf/stoc/ChenJK23}.
One property that we often need, beyond distortion of distances, is that $G$ has
a \emph{sub-Gaussian tail}, as described in \Cref{sec:technical_overview}.

\begin{theorem}[Main Result]
\label{thm:dim_reduction_informal}
For every $\alpha,d,k$ and $n$, 
there is a random linear map $G:\R^d\to\R^t$
with target dimension $t=O(\frac{\log n}{\alpha^2} + \log k)$,
such that
for every set $P\subset \R^d$ of $n$ points, with high probability,
$G$ preserves the $k$-center value of $P$ within $O(\alpha)$ factor.
This result extends to $k$-center variants 
as listed in~\Cref{table:dim_red_k_center_variants}.
\end{theorem}
\begin{remark}
We actually prove this theorem for target dimension $t=O(\tfrac{d}{\alpha^2}+ \log k)$, which is only stronger; indeed, \new{in \Cref{sec:JL_preserves_constrained_kcenter} we prove that by the JL Lemma and an extension theorem,}
one can assume that $d=O(\log n)$.
\end{remark}

\renewcommand{\arraystretch}{1.15}
\begin{table*}[!t]
\begin{center}
\caption{Dimension reduction bounds for $O(\alpha)$-estimation of $k$-center variants in terms of $D~=~\min(d,\log n)$.
The lower bound (marked by *) holds only for a matrix of iid Gaussians.
\label{table:dim_red_k_center_variants} 
}

\begin{tabular}{|l|c|l|}
\hline
$k$-center variant & target dimension & reference \\ 
\hline
\hline
all variants &  $O(D)$ & JL Lemma \\
  & $\Omega(\frac{D}{\alpha^2} + \frac{\log k}{\log \alpha})$ * & 
        \Cref{sec:lower_bound_gaussian_map}
    \\
  \hline
vanilla  & $O(\frac{D}{\alpha^2} + \log k)$ & \Cref{thm:dimension_reduction_for_kcenter} \\

\hline
with $z$ outliers  & $O(\frac{D}{\alpha^2} + \log(kz))$ & \Cref{thm:outliers_k_center} \\
\hline
assignment constrained & $O(\frac{D}{\alpha} + \log k)$ & \Cref{thm:constrained_k_center}  \\
\hline
\end{tabular}

\end{center}
\end{table*}

Our bound is nearly optimal
when $G$ is a matrix of iid Gaussians (and therefore, it plausibly holds for all JL maps).
Concretely, even for $k=1$ and the diameter (which are within factor $2$ of each other), the target dimension must be $t=\Omega(\tfrac{\log n}{\alpha^2})$, which matches the leading term in our bound; and
the second term is nearly matched by an $\Omega(\tfrac{\log k}{\log \alpha})$ bound.
For more details, see 
    \Cref{sec:lower_bound_gaussian_map}.

To demonstrate that moderate dimension reduction
is a general approach that may be applied more broadly, 
our \Cref{thm:dim_reduction_informal} includes 
non-trivial extensions to several important variants of $k$-center. 
From here on, we call the classic $k$-center mentioned above the ``vanilla'' variant, essentially to distinguish it from other variants that we now discuss.
In the variant \emph{with outliers} (aka robust $k$-center), 
the input specifies also $z\ge 0$,
and the goal is to find, in addition to the set of centers $C$,
a set of at most $z$ outliers $Z\subseteq P$, 
that minimizes the objective $\max_{p\in P\setminus Z}\dist(p,C)$
\cite{DBLP:conf/soda/CharikarKMN01}. 
\new{
Two other variants are the capacitated variant,
where every cluster has a bounded capacity~\cite{Bar-IlanKP93,DBLP:journals/siamdm/KhullerS00},
and the fair variant, where input points have colors 
and the relative frequency of colors in every cluster must be similar
to that of the entire input~\cite{DBLP:conf/nips/Chierichetti0LV17}.
These two variants are special cases of $k$-center \emph{with an assignment constraint}, which asks to assign each input point to one of the $k$ clusters
(not necessarily to the closest center), 
given a constraint on the entire assignment,
and the goal is to minimize the maximum cluster radius.
More formally, an \emph{assignment} is a map $\pi:[n]\to [k]$, and
an \emph{assignment constraint} is a partition of all possible assignments
into feasible and infeasible ones, formalized as $\AC:[k]^n\to \{0,1\}$.\footnote{To view a partition of an $n$-point dataset $P$ to $k$ clusters as an assignment,
represent $P$ by $[n]$ and the clusters by $[k]$, 
in an arbitrary manner (not by the geometry of the points).}
\begin{definition}
    In \emph{$k$-center with an assignment constraint $\AC$},
    the input is a set $P\subset \R^d$ of $n$ points,
    and the goal is to partition $P$ into $k$ sets (called clusters)
    in a manner feasible according to $\AC$ when viewed as an assignment,
    so as to minimize the maximum cluster radius.
\end{definition}
This formulation via a constrained assignment
has been studied before for other clustering problems~\cite{DBLP:conf/waoa/SchmidtSS19, DBLP:conf/nips/HuangJV19, DBLP:conf/icalp/BandyapadhyayFS21, DBLP:conf/focs/BravermanCJKST022}
but not for $k$-center.
It is useful as a generalization that captures both the capacitated variant
and the fair variant.
}
Perhaps surprisingly, the dimension reduction satisfies a strong ``for all'' guarantee with respect to the assignment constraints,
i.e., with high probability the optimal value is preserved simultaneously for all possible constraints $\AC:[k]^n\to \{0,1\}$.
This is in contrast to the weaker ``for each''  guarantee, 
where the high-probability bound applies separately to each constraint $\AC$.
The extensions to $k$-center variants 
come at the cost of a slightly increased target dimension,
as listed in~\Cref{table:dim_red_k_center_variants}.

Another important feature of our approach is that it actually preserves solutions,
i.e., \Cref{thm:dim_reduction_informal} extends
from $O(\alpha)$-estimation to $O(\alpha)$-approximation, as follows. 
For vanilla $k$-center, \new{a solution of $P$
is a set $C\subset \R^d$ of size $k$,
with cost (objective value) $\max_{p\in P} \dist(p,C)$.
Generally, a solution is called \emph{$\alpha$-approximation} if its cost 
is within factor $\alpha$ from the optimal.}
The proof of \Cref{thm:dim_reduction_informal} can be extended 
to show that a set $C\subseteq P$ is an $O(\alpha)$-approximate solution of $P$
whenever $G(C)$ is an $O(1)$-approximate solution of $G(P)$.\footnote{For a set $X\subseteq \R^d$ and matrix $G\in\R^{t\times d}$, denote $G(X) = \{Gx:\ x\in X\}$.}
The above restriction to $C\subseteq P$
is needed in our proof to guarantee that $C\mapsto G(C)$ is one-to-one, 
but one can then easily relax it to our true requirement $C\subset\R^d$
by introducing a factor of $2$ in the approximation ratio. 
A similar extension to $O(\alpha)$-approximation holds also for the aforementioned variants of $k$-center,
\new{
although the definitions for solutions slightly differ.
In the variant with $z$ outliers, a solution is a set $C\subset\R^d$ of $k$ centers and a set $Z\subseteq P$ of $z$ outliers, with cost $\max_{p\in P\setminus Z} \dist(p,C)$.
A solution to the variant with an assignment constraint is a feasible partition of $P$ into $k$ sets (notice it is \emph{not} based on centers).
}

\medskip

Still, an important question remains --- how useful is this theorem, 
and more generally, this regime of moderate dimension reduction?
It may seem that for fixed $\alpha>1$,
the target dimension $\tfrac{\log n}{\poly(\alpha)} = \Theta(\log n)$
offers only negligible improvement over the JL Lemma.
The crux is that many algorithms depend exponentially in the dimension,
in which case decreasing the dimension \emph{by a constant factor}
amounts to a \emph{polynomial improvement} in efficiency.
We indeed show such an application to streaming algorithms, as discussed next.

\subsection{Application: Dynamic Geometric Streams}\label{sec:dynamic_streams}

In dynamic geometric streams, a model
introduced by Indyk~\cite{10.1145/1007352.1007413_Indyk04},
the input $P$ is a set of points from $[\Delta]^d = \{1,2,...,\Delta\}^d$,
presented as a stream of point insertions and deletions. 
Algorithms in this model read the stream in one pass
and their space complexity (aka storage requirement) is limited.
We assume that $\Delta \leq \poly(n)$, which is common in this model. 
Ideally, algorithms for $k$-center should use at most $\poly(kd\log n)$ bits of space,
which is polynomial in the bit representation of a solution consisting of $k$ points in $[\Delta]^d$.
We focus throughout on the general case of high dimension (say $d\geq \log n$),
and mostly ignore algorithms 
whose space complexity is $\geq 2^d$ bits, which are suitable only for low dimension.

For $k$-center in insertion-only streams (i.e., without deletions), 
the tradeoff between approximation and space complexity is well understood,
and the regime of $O(1)$-approximation seems to be the most useful and interesting. 
\new{Indeed, there is an $O(1)$-approximation algorithm for discrete metrics that stores $O(k)$ points}
\cite{DBLP:journals/siamcomp/CharikarCFM04},
with extensions to the outliers variant~\cite{DBLP:conf/approx/McCutchenK08}
and also to the sliding-window model~\cite{CSS16}. 
\new{In Euclidean space, these algorithms yield $O(1)$-approximation using $O(kd\log n)$ bits of space.}
In contrast, all known $(1+\epsilon)$-approximation \new{algorithms in Euclidean space}
have space bound that grows like $(1/\epsilon)^d$~\cite{DBLP:journals/algorithmica/AgarwalP02,CeccarelloPP19,dBMZ21,BergBM23},
which is not sublinear in $n$ for high dimension. 
Furthermore, approximation below $\tfrac{1+\sqrt{2}}{2}\approx 1.207$ provably requires
$\Omega(\min\{n,\exp(d^{1/3})\})$ bits of space, even for $k=1$~\cite{AS15} (and is nearly matched by a $1.22$-approximation algorithm using $O(d\log n)$ bits~\cite{AS15,ChanP14,HMV25}). 
Another indication is a $(1.8+\epsilon)$-approximation for vanilla $k$-center using space complexity bigger by factor $k^{O(k)}$~\cite{DBLP:journals/comgeo/KimA15}.

The setting of dynamic streams (i.e., with deletions) seems much harder.
Here, $O(1)$-estimation using $\poly(kd \log n)$ bits of space is widely open, 
as we seem to lack effective algorithmic techniques. For example, in insertion-only streams,
$2$-approximation can be easily achieved for $k=1$, by just
storing the first point (in the stream) and then the furthest point from it.
It is unknown whether this algorithm extends to dynamic streams, 
because it relies on access to a point from $P$ right as the stream begins.
We provide an algorithm for dynamic streams and all $k$,
by simply applying our moderate dimension reduction result
and then using an algorithm of~\cite{BergBM23} for low dimension.

\begin{theorem}[Streaming Algorithm for $k$-Center]
\label{thm:streaming_informal}
There is a randomized algorithm that,
given $\alpha,d,k,n$ and a set 
    $P\subset \R^d$ of size at most $n$ presented as a stream of $\poly(n)$ insertions and deletions of points,
    returns an $O(\alpha)$-approximation to the $k$-center problem on $P$. 
    The algorithm uses 
    $n^{1/\alpha^2}\poly(kd \log n)$
    bits of space,
and extends to $k$-center variants as listed in~\Cref{table:dynamic_streams}.
\end{theorem}

\begin{remark}
  This theorem can also achieve $2^{d/\alpha^2} \poly(kd\log n)$ bits of space,
  which is only stronger; indeed,
  one can assume that $d=O(\log n)$ by the JL Lemma. 
  By setting $\alpha$ appropriately,
  this algorithm achieves $O(\sqrt{d/\log (kd\log n)})$-estimation
  using the ideal space bound of $\poly(kd\log n)$ bits.
\end{remark}

\begin{table*}[!t]
\begin{center}
  \caption{Space complexity upper bounds for $O(\alpha)$-estimation in dynamic streams of $k$-center variants,
    listed separately as a function of $d$ and of $n$, 
    but omitting $\poly(kd\log n)$ factors.
    The results of~\cite{BergBM23} actually achieve $(1+\epsilon)$-estimation.
    \label{table:dynamic_streams} 
} 
\begin{tabular}{|c|c|c|c|}
\hline
$k$-center variant & \multicolumn{2}{c|}{dynamic streaming space upper bounds} & reference \\ 
\hline
\hline
vanilla  & $2^{d}$  & - & \cite{BergBM23} \\
    & - & $n^{1/\alpha^2}$, only for $k=1$ & \cite{DBLP:conf/soda/Indyk03} \\
    & $2^{d/\alpha}$ & $n^{1/\alpha}$ & derived from \cite{DBLP_arxiv:CJK+22} \\
   & $2^{d/\alpha^2}$ & $n^{1/\alpha^2}$
& \Cref{thm:streaming_vanilla} \\
   \hline
with $z$ outliers & $2^{d} + z$ & - & \cite{BergBM23} \\
     & $2^{d/\alpha^2}\poly(z)$ & $n^{1/\alpha^2}\poly(z)$
& \Cref{thm:streaming_outliers} \\
     \hline
capacitated  & $2^{d}$  & - & \cite{BergBM23}\\
    & $2^{d/\alpha}$ & $n^{1/\alpha}$
& 
        \Cref{thm:streaming_capacitated} 
    
    \\
    \hline
fair  & $2^{d}$  & - & \cite{BergBM23}\\
    & $2^{d/\alpha}$ & $n^{1/\alpha}$ 
& 
        \Cref{thm:streaming_fairness} 
    
    \\
\hline
\end{tabular}
\end{center}
\end{table*}

We do not know whether our bounds are tight, and in fact,
proving lower bounds for dynamic geometric streams is a challenging open problem
(apparently, even for deterministic algorithms).
One would expect these lower bounds to exceed those known for insertion-only streams,
however, current techniques seem unable to exploit deletions.

Our result improves and significantly generalizes the known bounds for $k$-center in dynamic streams.
Currently, the best algorithm that can handle deletions
works only for the special case of $k=1$ vanilla variant, 
and achieves $O(\alpha)$-estimation using $\tO(n^{1/\alpha^2}d)$ bits of space\footnote{Throughout, the notation $\tO(f)$ hides $\poly(\log n)$ factors.}
(i.e., same bound as in \Cref{thm:streaming_informal}, but only for this special case).
This result follows by adapting a dynamic algorithm from~\cite{DBLP:conf/soda/Indyk03}
to the streaming setting. 
Known algorithms for insertion-only streams are largely not relevant,
as for this problem they rarely extend to handle deletions. 
Perhaps the only exception is a simple approach based on
moving the data points to a grid of a certain granularity
(see e.g.~\cite{DBLP:journals/algorithmica/AgarwalP02}),
that has been extended to handle deletions by employing sparse-recovery techniques~\cite{BergBM23}.
It was further extended to the outliers variant~\cite{BergBM23}, 
and it extends to the capacitated and fair variants as well, see \Cref{sec:streaming_low_dim} for a brief discussion.
This algorithm achieves $(1+\epsilon)$-approximation
and uses space complexity that is bigger by factor $(O(\frac{1}{\epsilon}))^d$,
which, as mentioned earlier, is not sublinear in $n$ for high dimension. 
Another possible approach is to employ a recent technique
called consistent hashing~\cite{DBLP_arxiv:CJK+22},
to obtain (quite easily, details omitted)
a streaming algorithm for vanilla $k$-center
that uses $2^{d/\alpha}\poly(kd \log n)$ bits of space. 
The dependence on $\alpha$ here is inferior to \Cref{thm:streaming_informal},
and is known to be optimal for consistent hashing~\cite{DBLP_arxiv:CJK+22}. 
We remark that the tree-embedding technique of \cite{10.1145/1007352.1007413_Indyk04},
which has been useful for several dynamic streaming problems in high dimension,
is ineffective for $k$-center, or even the diameter problem, 
as it bounds only the \emph{expected stretch} of every pair of points.

\new{
For dynamic streams in high dimension, while there is no prior work on $k$-center (except for $k=1$), there is work on other problems.
Algorithms achieving $(1 + \epsilon)$-approximations
using ``ideal'' space $\poly(\epsilon^{-1}kd \log \Delta)$
are known for clustering problems, like $k$-means \cite{hu2018nearly}
and $k$-median~\cite{DBLP:conf/icml/BravermanFLSY17},
and more recently for Max-Cut~\cite{DBLP:conf/stoc/ChenJK23}.\footnote{\new{For Max-Cut, which does not have parameter $k$, the ``ideal'' space complexity is $\poly(\epsilon^{-1}d \log \Delta)$.}}
For many other problems,
existing algorithms provide worse approximations,
like $O(d\log \Delta)$ or $O(d)$ 
\cite{10.1145/1007352.1007413_Indyk04, AIK08, DBLP:conf/focs/AndoniBIW09, DBLP:conf/stoc/ChenJLW22, WY22, DBLP_arxiv:CJK+22, DBLP:conf/stoc/ChenCJLW23}.
}

In fact, recent research on streaming algorithms
has uncovered an intriguing tradeoff between approximation and space complexity 
for different geometric problems.
Interestingly, the tradeoff we obtain for vanilla $k$-center
is better than the one known for these other problems.
For earth mover's distance (EMD) 
in the plane (i.e., $\R^2$), we know of $O(\alpha)$-estimation
using $\tO(n^{1/\alpha})$ bits of space~\cite{DBLP:conf/focs/AndoniBIW09}.\footnote{This result is in fact in terms of $\Delta$, but our assumption $\Delta \leq \poly(n)$ implies $\Delta^{O(1/\alpha)}=n^{O(1/\alpha)}$.}
For minimum spanning tree (MST), we know of $O(\alpha)$-estimation
using $n^{\sqrt{\nicefrac{\log \alpha}{\alpha}}}\poly(d\log n)$ bits of space~\cite{DBLP:conf/stoc/ChenCJLW23}.
For facility location, we know of $O(\alpha)$-estimation
using $n^{1/\alpha}\poly(d\log n)$ bits of space~\cite{DBLP_arxiv:CJK+22}.
We currently have no satisfactory explanation for these gaps, 
but since these four results rely on rather different methods, 
developing unified techniques (possibly via dimension reduction) may potentially improve some of these bounds.

\subsection{Extension: Inputs of Small Doubling Dimension}

When the input $P\subset\R^d$ has small doubling dimension,
one can achieve even better dimension reduction than \Cref{thm:dim_reduction_informal},
eliminating the dependence (of $t$) on $d$ and $n$.
Following~\cite{GKL03} (see also~\cite{Clarkson99}), 
the \emph{doubling dimension} of a set $P\subset\R^d$, denoted $\ddim(P)$,
is the smallest number such that every ball (in $P$)
can be covered by $2^{\ddim(P)}$ balls of half the radius.
Dimension reduction for inputs of small doubling dimension
has been studied before for three problems: 
For facility location, one can achieve $O(1)$-estimation 
using target dimension $t=O(\ddim(P))$~\cite{DBLP:conf/icml/NarayananSIZ21}.
For Nearest-Neighbour Search (NNS), one can obtain $(1+\epsilon)$-approximation
using $t=O(\tfrac{\log (1/\epsilon)}{\epsilon^2} \ddim(P))$~\cite{DBLP:journals/talg/IndykN07},
and for minimum spanning tree (MST), one can achieve $(1+\epsilon)$-estimation
using a similar target dimension albeit with another additive term of
$O(\tfrac{\log\log n}{\epsilon^2})$~\cite{DBLP:conf/icml/NarayananSIZ21}.
\new{Following our work, new bounds for these and other problems were obtained in~\cite{GJKSSSW25,HJKY25,JKSSY26}.}
The following theorem, which we prove in 
    \Cref{sec:doubling},
shows an analogous result for $(1+\epsilon)$-estimation of $k$-center.

\begin{theorem}[Dimension Reduction for Doubling Sets; abridged version of \Cref{thm:doubling_diameter}]
\label{thm:doubling_informal}
For every $\epsilon,d,k$ and $\ddim$, 
suppose $G$ is as in \Cref{sec:main_results}
with target dimension $t=O(\tfrac{\log(1/\epsilon)}{\epsilon^2}\ddim + \tfrac{\log k}{\epsilon^2})$.
Then, for every set $P\subset \R^d$ whose doubling dimension is at most $\ddim$, with high probability,
$G$ preserves the $k$-center value of $P$ within $1+\epsilon$ factor.\footnote{The $k$-center value here refers to $C\subset \R^d$,
  i.e., centers from the ambient space. 
  The theorem extends to preserving solutions,
  albeit with the restriction $C\subseteq P$,
  which introduces a factor of $2$ in the approximation. 
}
This result extends to the $k$-center variants listed in~\Cref{table:dim_red_k_center_variants}.
\end{theorem}

Observe that by composing the maps in \Cref{thm:dim_reduction_informal,thm:doubling_informal},
we can achieve $O(\alpha)$-approximation of $k$-center
when reducing to target dimension $t=O(\tfrac{\ddim(P)}{\alpha^2} + \log k)$.
This bound is better than the one in \Cref{thm:dim_reduction_informal}
since always $\ddim(P)=O(\min(d,\log n))$. 
Consequently,
the space of the streaming algorithm in \Cref{thm:streaming_informal}
improves to $2^{\ddim(P)/\alpha^2}\poly(kd\log n)$ bits.
The result extends to the $k$-center variants
(with outliers and with assignment constraint) in a natural way.
We can replace $d$ with $\ddim(P)$ also in \Cref{table:dim_red_k_center_variants,table:dynamic_streams}.
\new{The result for the assignment constraint variant holds ``for all'' constraints simultaneously with high probability, as we obtain in the analogous of \Cref{thm:dim_reduction_informal}.}

\section{Technical Overview}
\label{sec:technical_overview}
\label{SEC:TECHNICAL_OVERVIEW} \new{
In this section, we overview the proof of \Cref{thm:dim_reduction_informal},
including its extensions to the outliers and assignment constraint variants,
briefly overview the proofs of \Cref{thm:streaming_informal,thm:doubling_informal}, and additionally provide in \Cref{sec:all_pairs_moderate_dim} an $O(\tfrac{\log n}{\log \alpha})$ target dimension bound for $\alpha$-approximating all pairwise distances.
We focus on dimension reduction that maps
from dimension $d$ to $t=O(\frac{d}{\alpha^2}+\log k)$,
which implies \Cref{thm:dim_reduction_informal},
because we can effectively assume $d=O(\log n)$.
Indeed, if $d$ is larger then we first reduce it to $O(\log n)$
using \Cref{thm:doubling_informal} with $\eps=1/2$
and the bound $\ddim(P)\leq \log n$.
An alternative (and simpler) way to reduce a larger dimension $d$ to $O(\log n)$
is to use the JL Lemma and an extension theorem.
This is presented in \Cref{sec:JL_preserves_constrained_kcenter},
which furthermore shows that mapping to dimension $t=O(\eps^{-2}\log n)$ preserves all solutions, for all considered variants, up to $1+\eps$ factor.
(This bound has a slightly better dependence on $\eps$ compared to \Cref{thm:doubling_informal}, basically removing a $\log\tfrac{1}{\eps}$ factor.)
}

Our dimension-reduction map is defined via a matrix $G\in\R^{t\times d}$ of iid Gaussians $N(0,\tfrac{1}{t})$,
which is known to be a JL map~\cite{DBLP:conf/stoc/IndykM98,DBLP:journals/rsa/DasguptaG03}.
\new{Recall that a Gaussian $N(0,\tfrac{1}{t})$ is distributed like a Gaussian $N(0,1)$ scaled down by $\sqrt{t}$.}
\new{We remark that there are also other random linear maps that satisfy the JL Lemma, 
and our results may possibly be extended to these maps. 
These include projection to a random subspace~\cite{article_JohnsonLindenstrauss,FM90}, a matrix of iid Rademacher random variables~\cite{DBLP:journals/jcss/Achlioptas03} or 
independent random variables with a sub-Gaussian tail (which includes both Gaussians and Rademacher random variables)~\cite{DBLP:journals/rsa/Matousek08,DBLP:journals/talg/IndykN07,KLARTAG2005229}.
Moreover, there is a long line of work on maps with improved running time, e.g., the two cornerstone results known as fast JL~\cite{DBLP:conf/stoc/AilonC06} and sparse JL~\cite{DBLP:conf/stoc/DasguptaKS10}.}

\begin{lemma}[JL Lemma \cite{DBLP:conf/stoc/IndykM98,DBLP:journals/rsa/DasguptaG03}]
  \label{fact:JL_lemma_one_vector}
  Let $G$ be a $t\times d$ matrix of iid Gaussians $N(0,\tfrac{1}{t})$. Then
  \[
    \forall x\in \R^d \enspace \text{and} \enspace \epsilon>0,
    \qquad
    \Pr\big[ \|Gx\|\notin (1\pm \epsilon)\|x\| \big]
    \leq 2^{-\Omega(\epsilon^2 t)}.
  \]
  When $t\geq c_0 \eps^{-2}\log n$ for an absolute constant $c_0>0$, the right hand side is $\leq n^{-10}$. Therefore,
  for each set $P\subset\R^d$ of size $|P|=n$, by a union bound over pairs of points in $P$, with high probability,
  \[
  \forall x,y\in P, \quad \|Gx-Gy\|\in (1\pm\eps)\|x-y\|.
  \]
\end{lemma}
In addition to the above JL Lemma about distortion of distances,
we often use that a matrix of iid Gaussians satisfies
the following sub-Gaussian tail. 
\begin{definition}
    We say that a random map $f:\R^d \to \R^t$ has \emph{sub-Gaussian tail} if
    \begin{equation} \label{eq:subGaussian_tail}
      \forall x\in\R^d, r>0,
      \qquad
      \Pr_f \big[ \|f(x)\| \geq (1+r)\|x\| \big] \leq e^{-\Omega(r^2 t)}.
    \end{equation}
\end{definition}
This tail bound was key to prior work in this context,
and it holds for a $t\times d$ matrix of iid Gaussians $N(0,\tfrac{1}{t})$~\cite{DBLP:conf/stoc/MakarychevMR19, DBLP:conf/stoc/ChenJK23}.
The next technical lemma is key to our proof,
and shows that for a $t\times d$ matrix of iid Gaussians $N(0,1)$
(not normalized by $\tfrac{1}{\sqrt{t}}$),
with high probability, the largest singular value is $O(\sqrt{d})$.

\begin{lemma}\label{lem:ball_expansion_gaussians}
Let $t<d$ and suppose $\tilde{G}\in\R^{t\times d}$ is a matrix of iid Gaussians $N(0,1)$. 
Then for a suitable absolute constant $c_0>0$,
\[
  \Pr\big[ \sup_{\|x\| \leq 1} \|\tilde{G}x\| > c_0\sqrt{d} \big]
  \leq  2^{-\Omega(d)}.  
\]
\end{lemma}

This lemma can be derived from \cite[Theorem 4.6.1]{vershynin_2018},
which bounds by $O(\sqrt{d})$ the largest singular value
of a $d\times d$ matrix with independent sub-Gaussians entries,
and arguing that removing rows cannot increase the largest singular value.
For completeness, we provide in \Cref{sec:dimension_reduction_for_kcenter} another proof,
which may possibly be extended to other JL maps.

\subparagraph{Warm Up: the Furthest Point Query Problem.}
We start with moderate dimension reduction for the furthest point query (FPQ) problem,
which may be of independent interest.
In this problem, the input is a \emph{data set} $P\subset \R^d$ of size $|P|=n$ 
and a \emph{query set} $Q\subset \R^d$ of size $|Q|\leq k$,
and the goal is to report a point from $P$ that is furthest from the set $Q$. 
Let $FPQ_k(P,Q)$ denote this optimal value (distance from $Q$). 
As we show further below,
one can use this problem to achieve $2$-approximation for vanilla $k$-center,
by simply employing the famous Gonzalez's algorithm~\cite{DBLP:journals/tcs/Gonzalez85},
which essentially solves $k$ instances of FPQ. 
This FPQ problem admits the same dimension reduction as in \Cref{thm:dim_reduction_informal},
with a slightly simpler proof than for vanilla $k$-center, and thus serves as a good warm up.
The theorem and proof are provided rigorously in 
    \Cref{sec:fpq_appendix}.

\new{
Let $G\in\R^{t\times d}$ be a matrix of iid Gaussians $N(0,\tfrac1t)$,
where $t=\tfrac{d}{\alpha^2} + c_1\log k$ for an absolute constant $c_1>0$.
To show that $G$ preserves the $FPQ_k$ value of $(P,Q)$ within $O(\alpha)$ factor,
it suffices to prove that: 
\begin{enumerate}[label=(\roman*)]
    \item $FPQ_k(G(P),G(Q))\leq O(\alpha)FPQ_k(P,Q)$, and
    \item $FPQ_k(G(P),G(Q))\geq \Omega(FPQ_k(P,Q))$.
\end{enumerate}
To prove (i), observe that 
by \Cref{lem:ball_expansion_gaussians}, 
with high probability,
our map expands all vectors at most by factor $O(\sqrt{d/t})= O(\alpha)$,
thus the value of $FPQ_k(P,Q)$
increases at most by this factor. 
To prove (ii), consider a point $p^*\in P$ that is furthest from $Q$.
Now apply the JL Lemma with $\epsilon=\tfrac{1}{2}$ on $Q\cup\{p^*\}$,
whose size is at most $k+1$ 
and get that with high probability,
\[
  FPQ_k(G(P),G(Q))\geq\dist(Gp^*,G(Q))\geq \tfrac{1}{2}\dist(p^*,Q)=\tfrac{1}{2}FPQ_k(P,Q).
\] 
This concludes the proof for $FPQ$.
}

\subparagraph{Framework for Problems with Small Witness.}
An immediate corollary of \Cref{lem:ball_expansion_gaussians} is that for every $k$-center variant, if one uses target dimension $t\geq \tfrac{d}{\alpha^2}$, then the optimal value increases at most by factor $O(\alpha)$.
(Obviously, this fact may be useful for many other geometric problems.)
We denote by $\OPT(X)$ the optimal value of the problem at hand (e.g., vanilla or outliers) for a set $X$.

It remains to prove that the optimal value does not decrease much,
and for this we devise the following approach:
prove the existence of a small subset $P'\subseteq P$ (say, of size $O(k)$),
that we shall call a \emph{witness}, which satisfies $\OPT(P') = \Omega(\OPT(P))$,
and then apply the JL Lemma on this set, say, with $\epsilon=\tfrac{1}{2}$,
to get that with high probability, 
$G$ decreases all pairwise distances in $P'$ by at most factor $2$.
We then immediately get
(by restricting the centers to the dataset, which loses another factor $2$), 
\[
\OPT(G(P)) \geq \OPT(G(P')) 
\geq \tfrac{1}{4} \OPT(P') = \Omega(\OPT(P)),
\]
which concludes the proof.
We apply this witness-based approach below,
viewing it as a framework that may find additional uses in the future.
Our notion of ``witness'' is somewhat analogous to a coreset:
both notions preserve the cost in a certain way, and both have a small size.
Charikar and Waingarten~\cite{DBLP:journals/corr/CharikarWaingarten23}
relied on coresets in an analogous (but technically different) argument
to prove dimension reduction results for other clustering problems.

\subparagraph*{Witness for Vanilla $k$-Center.}
Consider running Gonzalez's algorithm (aka furthest-first traversal)~\cite{DBLP:journals/tcs/Gonzalez85},
which is the following iterative algorithm. 
Iteratively construct a set $P'\subseteq P$, initialized to contain one arbitrary point from $P$,
and then while $|P'|< k+1$, find a solution for $FPQ_k(P,P')$
(i.e., a point furthest from the current $P'$)
and add it to $P'$. \new{Eventually, $P'$ has size $k+1$.}
It is well known that the distance between the last point added to $P'$
and the earlier points is in the range $[\OPT_{vanilla}(P),2\OPT_{vanilla}(P)]$,
which implies that $\OPT_{vanilla}(P')\geq \tfrac{1}{2}\OPT_{vanilla}(P)$.
We use this set $P'$ as a witness in our result for vanilla $k$-center, see \Cref{sec:rand_proj_FPQ}.

\subparagraph*{Witness for $k$-Center with $z$ Outliers.}

For this variant, there is a witness of size $O(kz)$.
It follows from a ``peeling'' algorithm --- execute Gonzalez's algorithm $z+1$ times, and after each execution, delete from $P$ the $k+1$ points found in that execution.
This algorithm (and its proof of correctness) was already used in \cite{AbbarAIMahabadiV13} in the context of robust coresets, and is itself based on \cite{AgarwalHY08_peeling}.
We use this witness in our result for $k$-center with $z$ outliers, see \Cref{sec:outliers}.

\subparagraph*{Variant with an Assignment Constraint.}
We do not present a witness for this variant,
but rather bound the decrease in value via a different method.
Denote by $\OPT_{\AC}(\cdot)$ the value of $k$-center with an assignment constraint $\AC$.
Our proof in \Cref{sec:dim_red_assignment_constraint} compares $\OPT_{\AC}$ to $\OPT_{vanilla}$ (on the same input $P$) --- if these values are close,
then the proof follows from the fact that $\OPT_{vanilla}$ is preserved. 
Otherwise, $\OPT_{\AC}$ is significantly larger than $\OPT_{vanilla}$, and for the sake of analysis, we ``move'' every data point to its nearest ``vanilla center'', and get a weighted set of only $k$ points, whose total weight is $n$.
The crux of the proof shows that under the random linear map,
moving points of $P$ corresponds to moving points of $G(P)$,
which in turn does not change $\OPT_{\AC}(G(P))$ by too much,
essentially because by \Cref{lem:ball_expansion_gaussians},
the map $G$ keeps every point close to its nearest vanilla center.

\subparagraph*{Streaming Algorithm.}\label{par:technical_overview_streaming}
Our streaming algorithm in \Cref{thm:streaming_informal} is a corollary of the dimension reduction.
First apply the dimension reduction $G:\R^d\to \R^t$ of \Cref{thm:dim_reduction_informal} on every point in the input stream,
and then employ a known algorithm from~\cite{BergBM23}
whose space is exponential in the reduced dimension $t$,
namely, $k (O(\tfrac{1}{\epsilon}))^{t}\poly(\log n)$ bits.
In a nutshell, their algorithm (when applied in $\R^t$)
``moves'' every input point to its nearest grid point
in an $(\tfrac{\epsilon}{\sqrt{t}} \OPT_{vanilla})$-grid,
and uses sparse recovery to find all the non-empty grid points
(viewed as buckets of input points), which they call a relaxed coreset. 
Moving the points can clearly change the optimal value
additively by at most $\epsilon\cdot \OPT_{vanilla}$,
and the number of grid points within a ball of radius $\OPT_{vanilla}$
is bounded by roughly $(O(\tfrac{1}{\epsilon}))^t$,
leading to the above space requirement.\footnote{The proof in~\cite{BergBM23} has an extra factor of $\sqrt{t}^t$ in the space bound.
  This extra factor becomes $O(1)$ in their setting of a fixed dimension $t=O(1)$,
  but it is unnecessary in general by a volume argument that compares a Euclidean ball of radius $\OPT_{vanilla}$ to a cube of sidelength $\tfrac{\epsilon}{\sqrt{t}}\OPT_{vanilla}$. }

\new{In \Cref{sec:streaming_low_dim},} we adapt this algorithm to report a solution,
namely, centers from $\R^d$ (actually selected from the input $P$).
Observe that one cannot just map the centers found in low dimension $\R^t$
to points in high dimension $\R^d$ because there is no inverse map to $G$. 
Instead, we use a two-level $\ell_0$-sampler~\cite[Lemma 3.3]{DBLP_arxiv:CJK+22},
which can be viewed as a more sophisticated version of sparse recovery ---
add to each insertion/deletion of a grid point in $\R^t$
a ``data field'' containing the original input point in $\R^d$.
Now a two-level sampler will pick a random non-empty grid-point (bucket)
and then a random element from that bucket,
and it will report also the data field, revealing an original input point,
which can be used as a center point (up to factor $2$ in the approximation). 
This two-level sampler has space requirement bigger by only factor $d$,
essentially to store linear combinations of such data fields.
This method extends to the variants listed in \Cref{table:dynamic_streams},
where throughout (i.e., for all variants)
a reported solution refers to the set of $k$ centers.

\paragraph{Dimension Reduction for Sets of Low Doubling Dimension.} 
\new{
The proof of 
\Cref{thm:doubling_informal}, which is provided in \Cref{sec:doubling}, combines several established techniques, as follows.
By a standard net argument (see, e.g.,~\cite{GKL03}), there exists a subset $Y \subseteq P$ of size at most $k(2/\epsilon)^{\mathrm{ddim}(P)}$ such that $\mathrm{dist}(p, Y) \le \epsilon \cdot \mathrm{opt}_{\mathrm{vanilla}}(P)$ for all $p \in P$. We then apply the Johnson-Lindenstrauss (JL) Lemma to $Y$. Furthermore, by a result of Indyk and Naor~\cite{DBLP:journals/talg/IndykN07} regarding dimension reduction for sets with low doubling dimension, every point within an $(\epsilon \cdot \mathrm{opt}_{\mathrm{vanilla}}(P))$-ball centered at a point $y \in Y$ maps to a point within a $(6\epsilon \cdot \mathrm{opt}_{\mathrm{vanilla}}(P))$-ball around $Gy$. These two events suffice to guarantee that the dimension reduction holds simultaneously for all assignment constraints, as follows.
}

\new{
Consider an arbitrary constraint $\mathcal{C}$, and let us show that
(a) $\mathrm{opt}_{\mathcal{C}}(G(P)) \le (1 + O(\epsilon)) \mathrm{opt}_{\mathcal{C}}(P)$ and
(b) $\mathrm{opt}_{\mathcal{C}}(G(P)) \ge (1 - O(\epsilon)) \mathrm{opt}_{\mathcal{C}}(P)$;
the final guarantee then follows by rescaling $\epsilon$. 
To establish the upper bound (a), let $(C^*, \pi^*)$ be an optimal solution for $P$, where $C^* \subset \mathbb{R}^d$ is the set of centers and $\pi^*$ is a feasible partition with respect to $\mathcal{C}$.
The centers in $C^*$ need not belong to the dataset $P$,
but we can use an extension theorem of Kirszbraun~\cite{Kir34},
to extend the aforementioned JL map on $Y$
to a map $\phi$ over the larger domain $Y\cup C^*$, 
while maintaining a JL-like guarantee on this larger domain.
We then upper bound the clustering cost of $G(P)$ using the solution $(\phi(C^*),\pi^*)$ and a movement-based argument:
we ``move'' each point $p \in P$ to its nearest representative $y_p \in Y$,
apply the extended JL guarantee of $\phi$ on $Y \cup C^*$,
and then ``move'' the point back from $G(y_p)$ to $G(p)$.
The crux is that each movement incurs an additive cost of at most $6\epsilon \cdot \mathrm{opt}_{\mathrm{vanilla}}(P)$.}

\new{
The lower bound (b) follows by a similar argument, effectively reversing the roles of $P$ and $G(P)$.
Consider an optimal solution for $G(P)$, consisting of a set of centers $\tilde{C} \subset \mathbb{R}^t$ and a feasible partition $\tilde{\pi}$.
Since $G : Y \to G(Y)$ is a JL map,
its inverse map $G^{-1} : G(Y) \to Y$ approximately preserves distances,
and can be extended to the larger domain $G(Y) \cup \tilde{C}$.
Finally, use a similar movement-based argument to bound the cost of $P$ by
$\mathrm{opt}_{\mathcal{C}}(P) \le (1 + O(\epsilon)) \mathrm{opt}_{\mathcal{C}}(G(P))$,
which concludes the proof.
}

\new{
\subsection{$\alpha$-Approximation for All Pairwise Distances}\label{sec:all_pairs_moderate_dim}

We now provide a short proof that target dimension $t=O(\frac{\log n}{\log \alpha})$
suffices to preserve all the pairwise distances within factor $\alpha$. 

\begin{observation}[$\alpha$-approximation for all pairwise distances]\label{thm:all_pairs_moderate_dim}
Let $d,n$ and $\alpha>2$. There is a random linear map
$G : \mathbb R^d \to \mathbb R^t$ with
   $ t = O(\tfrac{\log n}{\log\alpha})$,
such that for every set $P\subset \mathbb R^d$ of size $n$, with probability at least $2/3$,
\[
    \forall p_1,p_2\in P,
    \qquad
    \|p_1-p_2\|\leq \|Gp_1-Gp_2\|
    \leq
        \alpha\|p_1-p_2\|.
\]
\end{observation}

The proof uses the following fact.
\begin{fact}[Equation 7 in \cite{DBLP:journals/talg/IndykN07}]
\label{fact:Gaussian_contraction}
    Let $ G\in \R^{t\times d}$ be a matrix of iid Gaussians $N(0,1/t)$. 
    Then 
    \[
    \forall x\in\R^d, r>0,
    \qquad
    \Pr \big[\|Gx\|>\|x\|/r\big] 
    \leq \big(\frac3r\big)^t.
    \]
\end{fact}

\begin{proof}[Proof of \Cref{thm:all_pairs_moderate_dim}]
Assume by the JL Lemma that $d=O(\log n)$.
Set
    $t=\lceil b_0\tfrac{\log n}{\log\alpha}\rceil$
for a sufficiently large absolute constant $b_0>0$, and let
$G$ be a $t\times d$ matrix of iid Gaussians $N(0,1/t)$. We first prove the desired guarantee for
$G$ up to an $\alpha'=\poly(\alpha)$ scaling, and then scale the map. (The change to $\alpha'$ only affects the hidden constant, since the target dimension depends on $\log \alpha$.)
To bound the expansion of $G$ observe that by \Cref{lem:ball_expansion_gaussians},
we have that with high probability, for all $x\in\R^d$, $\|Gx\|\leq \sqrt{d/t}\|x\|\leq O(\sqrt{\log \alpha}\|x\|)$. 
We will now bound the contraction of $G$.

Let $p_1,p_2\in P$. 
If $p_1=p_2$, then the desired inequality is trivial. 
Otherwise, 
by \Cref{fact:Gaussian_contraction},
\[
    \Pr\big[
        \|G(p_1-p_2)\|
        \leq
            \frac{1}{\sqrt{\alpha}}\|p_1-p_2\|
    \big]
    \le
    e^{-\Omega(t\log\alpha)} .
\]
By the choice of $t$, and by taking $b_0$ sufficiently large, the right hand side is at most
    $\frac{1}{10n^2}$.
Thus, by a union bound,
with probability $2/3$,
for all
$p_1,p_2\in P$,
\[
    \frac{1}{\sqrt{\alpha}}\|p_1-p_2\|\leq
    \|G(p_1-p_2)\|
    \leq
        O(\sqrt{\log\alpha})\|p_1-p_2\|.
\]
Rescaling $G$ we obtain distortion $\alpha'=O(\sqrt{\alpha}\log \alpha)$, concluding the proof.
\end{proof}
}

\section{Moderate Dimension Reduction for Vanilla $k$-Center}\label{sec:rand_proj_FPQ}\label{sec:dimension_reduction_for_kcenter}

In this section, we prove that a linear map via a matrix of iid Gaussians preserves the vanilla $k$-center value and solution, as follows.
This proves, in particular, the main claim in \Cref{thm:dim_reduction_informal}.

\begin{theorem}
\label{thm:dimension_reduction_for_kcenter}
Let $d,\alpha>1$ and $k\leq n$.
There is a random linear map $G:\R^d\to \R^t$
    with $t=O(\log k +\tfrac{d}{\alpha^2})$,
    such that 
    for every set $P\subset\R^d$ of size $n$,
    with probability at least $2/3$, 
    \begin{itemize}
\item $\OPT_{vanilla}(G(P))$ is an $O(\alpha)$-estimation for $\OPT_{vanilla}(P)$, and
        \item $C\subset P$ is an $O(\alpha)$-approximate vanilla $k$-center solution of $P$ whenever $G(C)$ is an $O(1)$-approximate vanilla $k$-center solution of $G(P)$.
    \end{itemize}
\end{theorem}
We start by providing a proof of \Cref{lem:ball_expansion_gaussians},
which is key to our proof of \Cref{thm:dimension_reduction_for_kcenter}. 
    We will need the following fact about Gaussians.

\begin{fact}\label{fact:chi_concentration}
    For every $t\geq 1$ and a Gaussian $g\sim N(0,I_t)$, 
    \[
      \forall r\geq\sqrt{5t}, \qquad
      \Pr\big[\|g\| \geq r \big] \leq e^{-r^2/5}.
    \]
\end{fact}
\begin{proof}
    By Laurent and Massart~\cite[Lemma 1]{10.1214/aos/1015957395_Laurent_Massart}, 
    \[
      \forall x>0, \qquad 
      \Pr\big[\|g\|^2-t\geq 2\sqrt{xt}+2x \big]\leq e^{-x}.
    \]
    Set $r^2=5x$, then
    $r^2\geq t + 2\sqrt{xt}+2x$,
    and thus
    \[
    \Pr\big[\|g\| \geq r \big] \leq \Pr\big[\|g\|^2-t\geq 2\sqrt{xt}+2x \big]\leq e^{-x}=e^{-r^2/5}.
    \]
\end{proof}
We will now prove \Cref{lem:ball_expansion_gaussians},
which states that for a $t\times d$ matrix $G$ of iid Gaussians $N(0,1)$ 
and a suitable absolute constant $c_0>0$, 
\[
  \Pr\big[ \sup_{\|x\| \leq 1} \|Gx\| > c_0 \sqrt{d} \big]
  \leq  2^{-\Omega(d)}.  
\]
\new{Throughout, the Euclidean ball of radius $r>0$ around $x\in\R^d$
is denoted by $B(x,r)=\set{y\in \R^d: \|x-y\|\leq r}$.
We denote the origin by $\vec{0}\in \R^d$, 
and thus $B(\vec{0},1)$ is the Euclidean unit ball.
}

\begin{proof}[Proof of \Cref{lem:ball_expansion_gaussians}.]
    Denote $T=c_0\sqrt{d}$.
    Define subsets $I_0,I_1,I_2,...\subset B(\vec{0},1)$ iteratively as follows.
    Let $I_0=\set{\vec{0}}$.
    For every $y\in I_j$, let $S_j(y)$ be a minimal set such that $\bigcup_{s\in S_j(y)}B(s,2^{-j-1})$ covers $B(y,2^{-j})$.
    Clearly, $\|y-s\|\leq 2^{-j}$.
    Let $I_{j+1}=\bigcup_{y\in I_j}S_j(y)$.
    It then holds that $|S_j(y)|\leq e^{O(d)}$ (see e.g. \cite{rogers_1963}).
    By induction, $|I_j|\leq e^{O(jd)}$. 
    
    For every $y\in I_j$ and $s\in S_j(y)$, 
    \begin{align*}
        \Pr\big[\|G(y-s)\|> \tfrac{T}{4}(\tfrac{4}{3})^{-j}\big]    &\leq  \Pr\big[\|G(y-s)\|> \tfrac{T}{4}(\tfrac{4}{3})^{-j}\|y-s\| 2^j\big] \\
        \intertext{and since $G\tfrac{y-s}{\|y-s\|}\sim N(0,I_t)$ we have by \Cref{fact:chi_concentration}, }
                                            &\leq e^{-\Omega((T(2/3)^{-j})^2)}
                                             = e^{-\Omega(d(2/3)^{-2j})}.
    \end{align*}
    It follows by a union bound that 
    \[
      \Pr\big[\exists y\in I_j,s\in S_j(y),\; 
        \|G(y-s)\|> \tfrac{T}{4}(\tfrac{4}{3})^{-j}\big]
      \leq |I_{j+1}|\cdot e^{-\Omega(d(2/3)^{-2j})}
      \leq e^{-\Omega(d (2/3)^{-2j})}.
    \]
    By another union bound, the probability that there exists $j\geq 0$ for which the above event occurs, is at most $\sum_{j=0}^\infty e^{-\Omega(d (2/3)^{-2j})}\leq e^{-\Omega(d)}$.
    If we assume this event does not occur,
    then every $x\in B(0,1)$ satisfies $\|G x\|\leq T$, as follows.
    Every $x\in B(0,1)$ can be expressed as $x=\sum_{j=0}^\infty (a_j-a_{j+1})$ where $a_{j}\in I_{j}$ and $a_{j+1}\in S_j(a_j)$. 
    Thus by the triangle inequality, $\|G x\|\leq \sum_{j=0}^\infty \|G(a_j-a_{j+1})\| \leq \sum_{j=0}^\infty \tfrac{T}{4}(\tfrac{4}{3})^{-j} = T$.
\end{proof}

    We are now ready to prove \Cref{thm:dimension_reduction_for_kcenter}.

\begin{proof}[Proof of \Cref{thm:dimension_reduction_for_kcenter}.]
Assume without loss of generality that 
$\OPT_{vanilla}(P)=1$.
Let $C^*\subset \R^d$ be a set of optimal centers, and let $S\subseteq P$ be an output of Gonzalez's algorithm~\cite{DBLP:journals/tcs/Gonzalez85} after $k+1$ steps, hence, $|S|=k+1$ and $\min_{s_1,s_2\in S} \|s_1-s_2\|\geq 1$.
We treat $S\cup C^*$ as a \emph{witness}, as described in \Cref{sec:technical_overview}.
Let $\epsilon=\tfrac{1}{2}$,
set $T=\Theta(\sqrt{d})$ with a hidden constant that satisfies \Cref{lem:ball_expansion_gaussians}, 
and set $t=O(\log k) + 400 \tfrac{T^2}{\alpha^2}$ 
such that the bound $2^{-\Omega(t)}$ in \Cref{fact:JL_lemma_one_vector}
is at most $\tfrac{1}{10(2k+1)^2}$.
Let $G$ be a $t\times d$ matrix of iid Gaussians $N(0,1)$.
Then, by \Cref{fact:JL_lemma_one_vector} and a union bound,
\begin{equation}\label{eq:JL_for_k_center}
  \Pr\big[\forall p_1,p_2\in S\cup C^*,\; 
        \|G(p_1-p_2)\|\in [1\pm 0.5]\sqrt{t}\|p_1-p_2\| \big]
    \geq 1-(2k+1)^2 2^{-\Omega(t)} 
    \geq \tfrac{9}{10}.
\end{equation}
By \Cref{lem:ball_expansion_gaussians},
\begin{equation}\label{eq:balls_bounded_expansion}
  \Pr\big[\forall x\in B(\vec{0},1),\;  \|Gx\|\leq T \big]
  \geq  1-2^{-\Omega(d)}. 
\end{equation} 
Assuming the events in \Cref{eq:balls_bounded_expansion,eq:JL_for_k_center} happen, the following holds. 

\subparagraph*{Value.}
For every set $C_G\subset \R^t$ of size $k$,
by restricting the pointset to the set $S$, and since $|S|=k+1 > |C_G|$,
\begin{align*}
\max_{p\in P} \dist(Gp,C_G) 
    &\geq \max_{s\in S} \dist(Gs,C_G) \geq \tfrac{1}{2} \min_{s_1,s_2\in S} \|G(s_1-s_2)\| \\
\shortintertext{and now by the event in \Cref{eq:JL_for_k_center} and by the choice of $S$,}
    &\geq \tfrac{\sqrt{t}}{4} \min_{s_1,s_2\in S} \|s_1-s_2\|
    \geq \tfrac{\sqrt{t}}{4} \geq \tfrac{T}{\alpha}.
\end{align*}
In the other direction, $\max_{p\in P} \dist(p,C^*)\leq 1$,
and thus by using $G(C^*)$ as a center set and
by the event in \Cref{eq:balls_bounded_expansion}, 
\begin{align*}
    \OPT_{vanilla}(G(P))
&\leq  \max_{p\in P} \dist(Gp,G(C^*)) \leq T.
\end{align*}
Scaling $G$ by a factor of $\tfrac{\alpha}{T}$ (i.e., the final map is $G'=\tfrac{\alpha}{T}G$) proves the first bullet.
For ease of presentation, we analyze the second bullet (about preserving solutions) without this scaling,
as the scaling factor cancels and does not affect the argument.

\subparagraph*{Solution.}
We want to show that $C\subseteq P$ is an $O(\alpha)$-approximate $k$-center solution of $P$ whenever $G(C)$ is a $2$-approximate $k$-center solution of $G(P)$.
We shall actually prove the contrapositive claim,
and consider a set $B\subseteq P$ for which there exists a point $p'\in P$ such that $\dist(p',B)> 5\alpha$.
Now for every point $p\in P$, denote by $c_p$ its closest center from $C^*$.
Then the following holds. By the triangle inequality,
\begin{align*}
\dist(Gp',G(B)) 
&=\min_{b\in B} \|G(p'-b)\| 
\geq \min_{b\in B} \|G(c_{p'}-c_b)\| - \|G(c_{p'}-p')\| - \|G(b-c_b)\| \\
\shortintertext{by the events in Equations \Cref{eq:balls_bounded_expansion} and \Cref{eq:JL_for_k_center},}
&\geq \min_{b\in B} \|G(c_{p'}-c_b)\| - 2T 
\geq \min_{b\in B} \tfrac{1}{2}\sqrt{t}\|c_{p'}-c_b\| - 2T \\
\shortintertext{by the triangle inequality,}
& \geq \min_{b\in B} \tfrac{1}{2}\sqrt{t}(\|p'-b\| - \|c_{p'}-p'\| - \|b-c_b\|) - 2T \\
\shortintertext{by our assumptions,}
&\geq  \tfrac{1}{2}\sqrt{t} (5\alpha - 2) - 2T
\geq 20 T.
\end{align*}
Thus, $G(B)$ is not a $2$-approximate $k$-center solution of $G(P)$.
This concludes the proof of \Cref{thm:dimension_reduction_for_kcenter}.
The constant $2$ in the approximation is arbitrary 
and could be changed to any other constant by adapting the other parameters.
\end{proof}

\new{

\subsection{On the Optimality of \Cref{thm:dim_reduction_informal}}\label{sec:lower_bound_gaussian_map}

Our bound for $\alpha$-estimation of vanilla $k$-center
is nearly optimal
when the dimension reduction is defined via a matrix $G$ of iid Gaussians (and plausibly for all JL maps).
We focus first on the leading term $O(\tfrac{\log n}{\alpha^2})$, and show that it is tight.
Consider the diameter problem, 
whose value is within factor $2$ of the $1$-center value.
By letting $P$ be points on a one-dimensional line,
one can see that the scaling factor of $G$ must be $\Omega(1/\sqrt{t})$.
Now let $P$ be an $\epsilon$-net of the unit sphere, say for $\epsilon=0.1$,
which can be realized with $d=\Theta(\log n)$.
Then with high probability, $G$ stretches some unit vector to length
$\Omega(\sqrt{{d}/{t}})$ \cite[Theorem 4.6.1]{vershynin_2018},
hence there is some pair of points in $G(P)$ whose distance is at least $\Omega(\sqrt{{d}/{t}})$. 
Since $G$ preserves the diameter, that distance is also bounded by $2\alpha$,
and altogether $t=\Omega(\tfrac{d}{\alpha^2})=\Omega(\tfrac{\log n}{\alpha^2})$.

Similarly, the second term in our bound is nearly optimal, 
namely, the target dimension must be $t=\Omega(\tfrac{\log k}{\log \alpha})$.
For ease of presentation, assume $G$ is a matrix of iid $N(0,\tfrac{1}{t})$, but we will have to scale it at the end.
On one hand, if $P$ is a set of $k$ orthonormal vectors and the origin,
then $\OPT_{vanilla}(P)=\tfrac{1}{2}$, and
with high probability, $\OPT_{vanilla}(G(P))\leq 2^{-\Omega(\tfrac{\log k}{t})}$, as follows.
The probability that none of the $k$ orthonormal vectors shrinks
to length at most $1/\beta$ is 
\[
  \prod_{i=1}^k \big(1-\Pr[\|G e_i\|\leq \tfrac{1}{\beta}]\big)
  = \big(1-\Pr[\|G_1\|\leq \tfrac{1}{\beta}]\big)^k
  \leq (1-(e\beta)^{-t})^k
  \approx 1-k(e\beta)^{-t},
\]
where $G_1$ denotes the first column of $G$.
Hence, if $\beta = 2^{\Omega(\tfrac{\log k}{t})}$, then, with high probability,
there is a vector among the $k$ orthonormal vectors whose length shrinks to at most $1/\beta$,
and thus $\OPT_{vanilla}(G(P))\leq \tfrac{1}{2\beta}$.
On the other hand, if $P$ is a set of points on a line,
then $G$ preserves all pairwise distances with high probability,
hence $\OPT_{vanilla}(G(P))=O(\OPT_{vanilla}(P))$.
Thus, to get $\alpha$-approximation (after scaling $G$ appropriately), it must be that $\alpha = \Omega(\beta)$, and hence $t=\Omega(\tfrac{\log k}{\log \alpha})$.

\subsection{Moderate Dimension Reduction for FPQ}\label{sec:fpq_appendix}
    
    Before proceeding into the $k$-center variants, we first provide a complete proof for FPQ.

\begin{theorem}\label{thm:FPQ_via_random_projection}
Let $d,\alpha>1$ and $k\leq n$.
There is a random linear map $G:\R^d\to \R^t$
    with $t=O(\log k +\tfrac{d}{\alpha^2})$,
    such that for every two sets,
    $P\subset\R^d$ of size $n$ and $Q\subset\R^d$ of size $k$,
    with probability at least $2/3$, \begin{itemize}
        \item the value $FPQ_k(G(P),G(Q))$ is an $O(\alpha)$-estimation for the value $FPQ_k(P,Q)$, and
        \item $p\in P$ is an $O(\alpha)$-approximate solution for $FPQ_k(P,Q)$ whenever $Gp$ is an $O(1)$-approximate solution for $FPQ_k(G(P),G(Q))$.
    \end{itemize}
\end{theorem}

\begin{proof}
Let $p^*\in P$ be a point that realizes $FPQ_k(P,Q)$, 
i.e., $\dist(p^*,Q) = \max_{p\in P} \dist(p,Q)$.
    Assume without loss of generality that $\dist(p^*,Q) = \alpha$, which can be obtained by scaling all the points in $P\cup Q$ by an $\tfrac{\alpha}{\dist(p^*,Q)}$ factor.
    Set $T=\Theta(\sqrt{d})$ such that the hidden constant satisfies the condition in \Cref{lem:ball_expansion_gaussians}, and
    set $t=O(\log k) + 400 \tfrac{T^2}{\alpha^2}$ and $\epsilon=\tfrac{1}{2}$, such that the bound in \Cref{fact:JL_lemma_one_vector} is $2^{-\Omega(t)}\leq \tfrac{1}{10k}$.
    Let $G$ be a $t\times d$ matrix of iid Gaussians $N(0,1)$.
    Then, by \Cref{fact:JL_lemma_one_vector} and a union bound,
    \begin{equation}\label{eq:JL_for_FPQ}
      \Pr\big[\forall q\in Q, \;
        \|G(p^*-q)\|\in [1\pm 0.5]\sqrt{t}\|p^*-q\| \big]
      \geq 1-k 2^{-\Omega(t)}
      \geq \tfrac{9}{10}.
    \end{equation}
    \Cref{eq:balls_bounded_expansion} holds as in the proof of \Cref{thm:dimension_reduction_for_kcenter}.
    If we assume the two events in \Cref{eq:balls_bounded_expansion,eq:JL_for_FPQ} happen, then 
    \[
    \dist(Gp^*,G(Q)) \geq \tfrac{1}{2}\sqrt{t} \dist(p^*,Q)=\tfrac{1}{2}\sqrt{t}\alpha \geq 10T,
    \]
    and for every $x\in\R^d$ for which $\dist(x,Q)\leq 1$,
    \[
    \dist(Gx,G(Q)) \leq T.
    \]
    Hence, for every $p\in P$, $\dist(Gp,G(Q)) \leq T \dist(p,Q) \leq \alpha T$.
    By scaling $G$, we get an $O(\alpha)$-estimation.
    
    Similarly, a $5$-approximate $FPQ_k(G(P),G(Q))$ is a point $Gp\in G(P)$ that must satisfy $\dist(Gp,G(Q)) \geq 2T$, hence it corresponds to a point $p\in P$ such that $\dist(p,Q) > 1 = \tfrac{1}{\alpha} \dist(p^*,Q)$.
    The constant $5$ in the approximation is arbitrary 
    and could be changed to any other constant by increasing the constants in $t$.
    This concludes the proof of \Cref{thm:FPQ_via_random_projection}.
\end{proof}

 }

 \section{Moderate Dimension Reduction for $k$-Center with Outliers}
\label{sec:outliers}
In this section, we design a moderate dimension reduction for $k$-center with $z$ outliers,
and demonstrate its application to streaming algorithms.
We denote by $\OPT_{outliers}(P)$ the optimal value of $k$-center with $z$ outliers of $P$, and say that $C,Z\subseteq P$ with $|Z|\leq z$ is an $\alpha$-approximate solution of $k$-center with $z$ outliers for $P$ if $\max_{p\in P\setminus Z} \dist(p,C)\leq \alpha \opt_{outliers}(P)$.
\begin{theorem}\label{thm:outliers_k_center}
Let $k,z\leq n$ and $d,\alpha>1$.
There is a random linear map $G:\R^d\to \R^t$
  with $t=O(\tfrac{d}{\alpha^2} + \log(kz))$,
  such that 
  for every set $P\subset\R^d$ of size $n$, 
  with probability at least $2/3$,
  \begin{itemize}
      \item $\OPT_{outliers}(G(P))\in \big[\OPT_{outliers}(P),O(\alpha)\cdot \OPT_{outliers}(P) \big]$, and
      \item $C,Z\subseteq P$ is an $O(\alpha)$-approximate solution of $k$-center with $z$ outliers for $P$ whenever $G(C),G(Z)$ is an $O(1)$-approximate solution of $k$-center with $z$ outliers for $G(P)$.
  \end{itemize}
\end{theorem}

Our proof uses a witness of size $O(kz)$, i.e., a subset $P'\subseteq P$ of size $|P'|=O(kz)$, such that $\OPT_{outliers}(P') = \Omega(\OPT_{outliers}(P))$.
\new{
This witness follows from known results for coresets that are robust to outliers~\cite[Corollary 4]{AbbarAIMahabadiV13}.
It is constructed by executing Gonzalez's algorithm $z+1$ times,
each time deleting the points returned by the previous execution.
An important property of this construction is that for every choice of $z$ outliers,
one of the $z+1$ executions of Gonzalez's algorithm returns a set of points without outliers.
We initially were not aware of~\cite{AbbarAIMahabadiV13},
and earlier versions of our paper (including the preliminary version in SoCG 2024~\cite{JiangKS24})
presented it as a new result. 
We prove the correctness of this witness for completeness.
}

\begin{lemma}[Witness for the outliers variant \cite{AbbarAIMahabadiV13}]\label{lem:witness_outliers}
    For every $k,z$ and $P\subset\R^d$, there is a subset $P'\subseteq P$ of size $|P'|=(k+1)(z+1)$, such that 
    $\OPT_{outliers}(P')\in [\tfrac{1}{3}\OPT_{outliers}(P),\OPT_{outliers}(P)]$.
\end{lemma}

\begin{remark}
    The size of the witness of \Cref{lem:witness_outliers} is tight up to low-order terms, by the following example.
    Consider a set $X$ of $(k+1)z+1$ points, where there are $k+1$ locations with pairwise distances $1$, such that each location contains $z$ points, and the last remaining point is at distance $\tfrac{1}{3}$ from one of these locations.
    Clearly, the cost of $k$-center with $z$ outliers is at least $\tfrac{1}{6}$,
    by considering the $z$ points from one of the locations as outliers.
    However, every strict subset of $X$ (i.e., excluding even one point) has cost $0$, as there will be a location with $z-1$ points, which can be taken as outliers, together with that last point.
\end{remark}

\begin{proof}[Proof of \Cref{lem:witness_outliers}]
For a set $X$, we denote by $Gonz(X,k+1)$ a set of $k+1$ points computed by executing Gonzalez's algorithm (for $k$ iterations) on $X$,
breaking ties (e.g., the initial point) arbitrarily. 
Given $P\subset\R^d$, construct a witness for $k$-center with $z$ outliers as follows.
\begin{enumerate}
    \item $X \gets P$
    \item for $i=1,\ldots,z+1$
    \item \qquad $C_i \gets Gonz(X,k+1)$
    \item \qquad $X\gets X\setminus C_i$
    \item return $P' = \cup_{i\in [z+1]} C_i$ as a witness
\end{enumerate}
Clearly, $\OPT_{outliers}(P')\leq \OPT_{outliers}(P)$. It remains to prove that $\OPT_{outliers}(P') \geq \tfrac{1}{3}\OPT_{outliers}(P)$.

Let $C'$ and $Z'$ be the optimal centers and outliers for $P'$, respectively.
Since $|Z'|\leq z$, there exists $i\in [z+1]$ such that $C_i\cap Z'=\emptyset$.
By the pigeonhole principle, there are $p_1,p_2\in C_i$ that are clustered to the same cluster by $C'$, and thus by the triangle inequality, $\dist(p_1,p_2)\leq 2\OPT_{outliers}(P')$. 
Suppose without loss of generality that $p_1$ was added to $C_i$ before $p_2$ in the execution of Gonzalez's algorithm.
We now upper bound $\OPT_{outliers}(P)$ by considering centers $C'$ and outliers $Z'$, and get
\begin{align*}
    \OPT_{outliers}(P)  
    &\leq \max_{p\in P\setminus Z'} \dist(p, C') 
    = \max \{ \max_{p\in P\setminus P'} \dist(p, C'),\max_{p\in P'\setminus Z'} \dist(p, C') \}. 
\end{align*}
The second term is by definition $\OPT_{outliers}(P')$, so let us bound the first term. 
For every $p\in P\setminus P'$,
by the triangle inequality,
\begin{equation} \label{eq:outliers1}
 \dist(p, C') \leq \min_{p'\in C_i} \{ \dist(p,p') + \dist(p', C') \}
\leq  \dist(p,C_i) + \OPT_{outliers}(P').
\end{equation}
Let $\hat{C}_i\subseteq C_i$ be the set $C_i$ \emph{at the time that $p_2$ was chosen} (in the $i$-th execution of Gonzalez's algorithm). Then, because $p\notin P'$ is available at this time, and since $p_1\in \hat{C}_i$,
\begin{align*}
    \dist(p,C_i)&\leq \dist(p,\hat{C}_i)  
    \leq \dist(p_2,\hat{C}_i) 
\leq \dist(p_2,p_1)  \leq 2\OPT_{outliers}(P'). 
\end{align*}
Together with \eqref{eq:outliers1}, we obtain 
$\dist(p,C') 
\leq 3\OPT_{outliers}(P')
$, 
which concludes the proof.
\end{proof}

\begin{proof}[Proof of \Cref{thm:outliers_k_center}]
    The proof that the value is preserved within factor $O(\alpha)$ using target dimension $t=O(\tfrac{d}{\alpha^2} + \log(kz))$ is the same as the proof for the vanilla variant, albeit with the witness given by \Cref{lem:witness_outliers}.
    As for the proof that solutions are preserved, it only requires few minor changes, as follows.

    Assume without loss of generality that $\OPT_{outliers}(P)=1$.
    Let $C^*$ and $Z^*$ be sets of optimal centers and outliers, respectively.
    We can assume without loss of generality that the points in $Z^*$ are furthest from $C^*$.
    Let $G$ be a $t\times d$ matrix of iid Gaussians $N(0,1)$, and
    set $t\geq b_0 \log(kz)$, where $b_0>0$ is an absolute constant such that by $\Cref{fact:JL_lemma_one_vector}$ and a union bound,
    \begin{equation}\label{eq:JL_outliers}
        \Pr\big[\forall p_1,p_2\in P'\cup C^*\cup Z^*,\; 
        \|G(p_1-p_2)\|\in [1\pm 0.5]\sqrt{t}\|p_1-p_2\| \big] 
        \geq \tfrac{9}{10}.
    \end{equation}
    For every point $p\in P\setminus Z^*$, denote by $c_p$ its closest center from $C^*$,
    and with slight abuse notation, for $p\in Z^*$ denote by $c_p$ the point $p$ itself.
The proof now considers sets $B\subseteq P$ of size $k$ and $Z'\subseteq P$ of size $z$ for which there exists a point $p'\in P\setminus Z'$ such that $\dist(p',B)>5\alpha$,
    and proceeds similarly to the proof of \Cref{thm:dimension_reduction_for_kcenter}.
\end{proof}

 \section{Moderate Dimension Reduction for $k$-Center with an Assignment Constraint}
\label{sec:dim_red_assignment_constraint}
\label{SEC:DIM_RED_ASSIGNMENT_CONSTRAINT} 

In this section, we consider $k$-center with an assignment constraint,
which captures the capacitated and fair variants of $k$-center,
as described in~\Cref{sec:main_results}. 
We design for this problem a moderate dimension reduction,
and demonstrate its application to streaming algorithms. 
Our definition below of an assignment constraint follows
the one used in~\cite{DBLP:conf/waoa/SchmidtSS19,DBLP:conf/nips/HuangJV19,DBLP:conf/icalp/BandyapadhyayFS21,DBLP:conf/focs/BravermanCJKST022}
for other $k$-clustering problems.
The \emph{radius} of a pointset is the optimal value of $1$-center clustering for it.

\begin{definition}\label{def:assigment_constraint}
An \emph{assignment} is a map $\pi:[n]\to [k]$.
An \emph{assignment constraint} is a partition of all possible assignments
into feasible and infeasible ones, formalized as $\AC:[k]^n\to \{0,1\}$.
\end{definition}
To view a partition of an $n$-point dataset $P$ to $k$ clusters as an assignment,
represent $P$ by $[n]$ and the clusters by $[k]$, 
in an arbitrary manner (not by the geometry of the points).
\Cref{def:assigment_constraint} can model clustering with capacity $L>0$,
by declaring an assignment $\pi$ to be feasible
if $|\pi^{-1}(i)|\leq L$ for all $i\in[k]$. 
To exemplify how it can model fair clustering,
suppose the first $n/3$ points in $P$ are colored blue and the others are red;
then declare $\pi$ to be feasible 
if in every $\pi^{-1}(i)$,
exactly $1/3$ of the elements 
are from the range $\{1,\ldots,\frac{n}{3}\}$.

\begin{definition}
In \emph{$k$-center with an assignment constraint $\AC$},
the input is a set $P\subset \R^d$ of $n$ points,
and the goal is to partition $P$ into $k$ sets (called clusters)
in a manner feasible according to $\AC$ when viewed as an assignment,
so as to minimize the maximum cluster radius.
The minimum value attained is denoted by $\OPT_{\AC}(P)$.
A \emph{solution} to this problem is a partition of $P$ into $k$ sets,
and it is called \emph{$\alpha$-approximate} for $\alpha\geq 1$ if it is feasible and
has value at most $\alpha\cdot \OPT_{\AC}(P)$.
\end{definition}

The next theorem shows that reducing to dimension $O(\tfrac{d}{\alpha}+\log k)$ preserves,
with high probability, the value of $k$-center with an assignment constraint up to an $O(\alpha)$ factor,
simultaneously for all assignment constraints.

\begin{theorem}\label{thm:constrained_k_center}
Let $d,\alpha>8$ and $k\leq n$.
There is a random linear map $G:\R^d\to \R^t$
with $t=O(\tfrac{d}{\alpha} + \log k )$,
  such that 
  for every set $P\subset\R^d$ of size $n$, with probability at least $2/3$, the following holds.
  For all $\AC:[k]^n\to \{0,1\}$,
  \begin{itemize}
      \item $\OPT_{\AC}(G(P))\in \big[\OPT_{\AC}(P),O(\alpha)\cdot  \OPT_{\AC}(P) \big]$, and
      \item \new{a feasible partition of $P$ to $k$ clusters and a set of centers $C\subseteq P$ are an $O(\alpha)$-approximate solution of $k$-center with assignment constraint $\AC$
        whenever the corresponding partition of $G(P)$ and set of centers $G(C)$ are an $O(1)$-approximate solution of $k$-center with assignment constraint $\AC$ for $G(P)$.}
  \end{itemize}
\end{theorem}

Our dimension bound here is worse than for the vanilla variant by factor $\alpha$,
essentially because our lower bound for the value of $G(P)$ is weaker.
More precisely, we take $t=O(\tfrac{d}{\alpha^2}+\log k)$ as before 
and let $G$ to be a matrix of iid Gaussians $N(0,\tfrac{1}{t})$,
however now we will prove that
$\tfrac{1}{\alpha} \OPT_{\AC}(P) \leq \OPT_{\AC}(G(P)) \leq \alpha \OPT_{\AC}(P)$.
While the upper bound here is as before,
the lower bound is weaker by factor $\alpha$,
hence we will have to scale $G$ appropriately and conclude the theorem for approximation $\alpha'=\alpha^2$.

\new{
Let us recall the proof idea provided in \Cref{sec:technical_overview}.
We do not present a witness for this variant,
but rather bound the decrease in value via a different method.
Our proof compares $\OPT_{\AC}$ to $\OPT_{vanilla}$ on the same input $P$;
if these values are close,
then the proof follows from the fact that $\OPT_{vanilla}$ is preserved.
Otherwise, $\OPT_{\AC}$ is significantly larger than $\OPT_{vanilla}$, and for the sake of analysis, we ``move'' every data point to its nearest ``vanilla center'', and get a weighted set of only $k$ points, whose total weight is $n$.
Moving points of $P$ corresponds to moving points of $G(P)$,
which in turn does not change $\OPT_{\AC}(G(P))$ by too much,
essentially because by \Cref{lem:ball_expansion_gaussians},
the map $G$ keeps every point close to its nearest vanilla center.
}
We will need the following lemma.
Throughout this section, all sets are multisets.
\begin{lemma}\label{lem:move_points_capacitated_k_center}
For every assignment constraint $\AC$, and for every set $X$ of $n$ points,
if the set $X'$ is constructed by moving every point in $X$ at most distance $\Delta>0$,
then $\abs{\OPT_{\AC}(X) - \OPT_{\AC}(X')} \leq \Delta $.
\end{lemma}
\begin{proof}
To see that $\OPT_{\AC}(X')\geq \OPT_{\AC}(X)-\Delta$,
consider an optimal solution for $X'$, increase its value (i.e., the radius of every cluster) by $\Delta$ and use the triangle inequality.
Now $\OPT_{\AC}(X)\geq \OPT_{\AC}(X')-\Delta$ follows by symmetry,
because $X$ can be obtained from $X'$ by moving every point at most distance $\Delta>0$.
\end{proof}

\begin{proof}[Proof of \Cref{thm:constrained_k_center}]
\new{
We first apply the same probability bounds as in \Cref{sec:dimension_reduction_for_kcenter}.
Let $C^*\subset \R^d$ be a set of $k$ centers that are optimal for the vanilla variant, and let $S\subseteq P$ be an output of Gonzalez's algorithm~\cite{DBLP:journals/tcs/Gonzalez85} after $k+1$ steps, hence, $|S|=k+1$ and $\min_{s_1,s_2\in S} \|s_1-s_2\|\geq \opt_{vanilla}(P)$.
Let $\epsilon=\tfrac{1}{2}$, 
set $T=\Theta(\sqrt{d})$ with a hidden constant that satisfies \Cref{lem:ball_expansion_gaussians}, 
and set $t=O(\log k) + 64 \tfrac{T^2}{\alpha^2}$ 
such that the bound $2^{-\Omega(t)}$ in \Cref{fact:JL_lemma_one_vector}
is at most $\tfrac{1}{10(2k+1)^2}$.
    Let $G$ be a $t\times d$ matrix of iid Gaussians $N(0,\tfrac{1}{t})$.
    Then, by \Cref{fact:JL_lemma_one_vector} and a union bound,
\begin{equation}\label{eq:JL_constraints}
  \Pr\big[\forall p_1,p_2\in S\cup C^*,\; 
        \|G(p_1-p_2)\|\in [1\pm 0.5]\|p_1-p_2\| \big]
    \geq 1-(2k+1)^2 2^{-\Omega(t)} 
    \geq \tfrac{9}{10}.
\end{equation}
By \Cref{lem:ball_expansion_gaussians},
\begin{equation}\label{eq:expansion_constraints}
\Pr\big[\forall x\in B(\vec{0},1),\;  \|Gx\|\leq \tfrac{\alpha}{8} \big]\geq
    \Pr\big[\forall x\in B(\vec{0},1),\;  \|Gx\|\leq \tfrac{T}{\sqrt{t}} \big] 
  \geq  1-2^{-\Omega(d)}. 
\end{equation} 
Assume henceforth that the events in \Cref{eq:JL_constraints,eq:expansion_constraints} happen. These events will suffice for us to obtain \Cref{thm:constrained_k_center} simultaneously for all constraints. }

Consider a constraint $\AC$. 
By taking an optimal solution for $P$
and applying \Cref{eq:expansion_constraints}, we obtain
    \[
    \OPT_{\AC}(4G(P)) = 4 \OPT_{\AC}(G(P)) \leq \tfrac{\alpha}{2} \OPT_{\AC}(P).
    \]
We shall show that $\OPT_{\AC}(4G(P)) \geq \tfrac{1}{\alpha}\OPT_{\AC}(P)$;
it will then follow that $\OPT_{\AC}(4\alpha G(P))$ is an $\alpha^2$-approximation of $\OPT_{\AC}(P)$,
completing the proof for the map $G'=4\alpha G$ and approximation $\alpha'=\alpha^2$.
Towards this goal, recall an argument from \Cref{sec:dimension_reduction_for_kcenter}:
since $S\subseteq P$ is a witness set, we have by \Cref{eq:JL_constraints} that
    \begin{align*}
        \OPT_{vanilla}(4G(P)) 
&\geq \OPT_{vanilla}(P).
\end{align*}
    Therefore, if $\OPT_{\AC}(P) \leq \alpha \OPT_{vanilla}(P)$, then
    \[
    \OPT_{\AC}(4G(P)) \geq \OPT_{vanilla}(4G(P)) \geq \OPT_{vanilla}(P) \geq \tfrac{1}{\alpha}\OPT_{\AC}(P) ,
    \]
    as desired, and hence
    we assume from now on that $\OPT_{\AC}(P) > \alpha \OPT_{vanilla}(P)$.
    Let $P'$ be the multi-set obtained by moving every point $p\in P$ to its nearest center $c^*_p\in C^*$. 
    Clearly, every point moves distance at most $\OPT_{vanilla}(P)$.
    Additionally, $G(P')$ equals to the set obtained by moving every point $Gp\in G(P)$ to its ``projected center'' $Gc^*_p$, and by \Cref{eq:expansion_constraints}, 
    $\|Gp-Gc^*_p\| \leq \tfrac{\alpha}{8}\|p-c^*_p\|\leq \tfrac{\alpha}{8}\OPT_{vanilla}(P)$, i.e., every point in $4G(P)$ moves distance at most $\tfrac{\alpha}{2}\OPT_{vanilla}(P)$.
\begin{align*}
  \OPT_{\AC}(4G(P)) 
  &\geq \OPT_{\AC}(4G(P')) - \tfrac{\alpha}{2}\OPT_{vanilla}(P)
  & \text{by \Cref{lem:move_points_capacitated_k_center}}
  \\
  &\geq 2\OPT_{\AC}(P') - \tfrac{\alpha}{2}\OPT_{vanilla}(P)
  & \text{by \Cref{eq:JL_constraints}}
  \\
  &\geq 2\big(\OPT_{\AC}(P) - \OPT_{vanilla}(P)\big) - \tfrac{\alpha}{2}\OPT_{vanilla}(P)
  & \text{by \Cref{lem:move_points_capacitated_k_center}}
  \\
  &\geq 2\big(\OPT_{\AC}(P) - \tfrac1\alpha \OPT_{\AC}(P)\big) - \tfrac{1}{2}\OPT_{\AC}(P)
  & \text{by our assumption}
  \\
  &\geq \tfrac{1}{2}\OPT_{\AC}(P)
  & \text{since $\alpha\geq 2$.}    
\end{align*}
    This concludes the proof of the value estimation.

    \new{
    To obtain the approximation guarantee (second bullet),
    consider a feasible partition $\mathcal{P}=(P_1,\ldots,P_k)$ of $P$ into $k$ clusters,
    and $k$ centers $C=(c_1,\ldots,c_k)\subseteq P$. 
    Define $\cost(\mathcal{P},C)\coloneqq\max_{i\in [k]}\dist(c_i,P_i)$. We now carry a similar argument, to show that $\tfrac{1}{\alpha}\cost(\mathcal{P},C)\leq \cost(4G(\mathcal{P}),4G(C))\leq \tfrac{\alpha}{2}\cost(\mathcal{P},C)$, and the proof will be concluded by rescaling $G$.
    First, by \Cref{eq:expansion_constraints}, $\cost(4G(\mathcal{P}),4G(C))\leq \tfrac{\alpha}{2}\cost(\mathcal{P},C)$.

    Next, if $\cost(\mathcal{P},C)\leq \alpha\opt_{vanilla}(P)$, then
    \[
    \cost(4G(\mathcal{P}),4G(C))\geq \opt_{vanilla}(4G(P))\geq \opt_{vanilla}(P)\geq \tfrac{1}{\alpha}\cost(\mathcal{P},C),
    \]
    as desired, hence we assume $\cost(\mathcal{P},C)>\alpha\opt_{vanilla}(P)$. Denote by $C'\subseteq P'$ a multiset obtained by moving every point $c\in C$ to its nearest optimal center $c^*_c\in C^*$.
    As shown above, every point in $4G(C)$ moves at most $\tfrac{\alpha}{2}\opt_{vanilla}(P)$.
    Thus,
    \begin{align*}
        \cost(4G(\mathcal{P}),4G(C)) &\geq \cost(4G(\mathcal{P'}),4G(C'))- \alpha\opt_{vanilla}(P)&& \text{by the triangle inequality} \\
        &\geq 2\cost(\mathcal{P'},C')- \alpha\opt_{vanilla}(P) && \text{by \Cref{eq:JL_constraints}} \\
        &\geq 2(\cost(\mathcal{P},C)-2\opt_{vanilla}(P))- \alpha\opt_{vanilla}(P) && \text{by the triangle inequality} \\
        &\geq 2(\cost(\mathcal{P},C)-\tfrac{2}{\alpha}\cost(\mathcal{P},C))- \cost(\mathcal{P},C) && \text{by our assumption}\\
        &\geq \tfrac12 \cost(\mathcal{P},C) && \text{since $\alpha\geq 8$,}
    \end{align*}
    and the proof is concluded by rescaling $G$.
    }
\end{proof}

    \section{Dimension Reduction for $k$-Center in Doubling Sets}
\label{sec:doubling}

In this section, 
we prove \Cref{thm:doubling_informal}.
More formally, we prove the following for $k$-center \new{with an assignment constraint. The vanilla version follows as a special case, and so does the outliers variant, by setting $k'=k+z$ and restricting $z$ of the clusters to have only one point.}

\begin{theorem}[Dimension Reduction for Doubling Sets]
  \label{thm:doubling_diameter}
    Let $P\subset\R^d$ be a set of doubling dimension $\ddim(P)$,
    let $k\leq |P|$ and $0<\epsilon<1/2$,
    and suppose $G\in\R^{t\times d}$ is a matrix of iid Gaussians $N(0,\tfrac{1}{t})$
    for suitable $t=O(\tfrac{\log k}{\epsilon^2}+\tfrac{ \log(1/\epsilon)}{\epsilon^2}\ddim(P))$.
    Then with probability at least $2/3$, \new{for all constraints $\AC$,}
    \begin{itemize}
        \item the $k$-center \new{with assignment constraint $\AC$} value of $G(P)$ is a $(1+\epsilon)$-approximation for the $k$-center \new{with assignment constraint $\AC$} value  of $P$, and
        \item \new{centers $C\subseteq P$ with assignment $\pi$ are an $(\alpha(1+\epsilon))$-approximate $k$-center with assignment constraint $\AC$ solution of $P$ whenever $G(C)$ with assignment $\pi$ are an $\alpha$-approximate $k$-center with assignment constraint $\AC$ solution of $G(P)$.}
    \end{itemize}
\end{theorem}

\begin{corollary}[Doubling Sets with Outliers]
  \label{thm:doubling_with_outliers}
    Under the conditions of \Cref{thm:doubling_diameter}, if $z<|P|$, then for suitable $t=O(\tfrac{\log (kz)}{\epsilon^2}+\tfrac{ \log(1/\epsilon)}{\epsilon^2}\ddim(P))$, the conclusion of \Cref{thm:doubling_diameter} holds for $k$-center with $z$ outliers.
\end{corollary}

\new{
Two key components in the proof of \Cref{thm:doubling_diameter} are
a lemma of Indyk and Naor~\cite{DBLP:journals/talg/IndykN07} and the Kirszbraun Theorem~\cite{Kir34},
which we state next.
As usual, a map $\phi:X\to Y$ is called \emph{$L$-Lipschitz}
(for $L\ge 1$ and subsets $X,Y\subseteq \R^d$ endowed with the $\ell_2$-norm)
if for all $x_1,x_2\in X$ we have $\|\phi(x_1)-\phi(x_2)\|\leq L\|x_1-x_2\|$.
}

\begin{lemma}[Lemma 4.2 in~\cite{DBLP:journals/talg/IndykN07}]
\label{lem:IN07_doubling_radius}
Let $X\subset B(\vec{0},1)$ be a subset of the Euclidean unit ball.
There are absolute constants $c,C>0$ such that for $t>C \ddim(X) +1$, and a matrix $G\in\R^{t\times d}$ of iid Gaussians $N(0,\tfrac{1}{t})$,
\[
\Pr (\exists x\in X, \|Gx\| >6) \leq e^{-ct}.
\]
\end{lemma}

\new{
\begin{theorem}[Kirszbraun Theorem~\cite{Kir34}]
\label{thm:kirszbraun}
For every subset $X\subset \R^t$ and an $L$-Lipschitz map $\phi:X\to \R^d$, there exists an $L$-Lipschitz extension $\tilde{\phi}$ of $\phi$ to the entire space $\R^t$.
\end{theorem}
}

\begin{proof}[Proof of \Cref{thm:doubling_diameter}]
Let $D=\opt_{vanilla}(P)$ be the vanilla $k$-center value of $P$.
Consider an $(\epsilon D)$-net for the set $P$,
i.e., a subset $Y\subseteq P$ such that for every $p\in P$ there exists $y\in Y$ satisfying $\|p-y\|\leq \epsilon D$.
By a standard argument (see~\cite{GKL03}),
there is such a net of size $\leq k(2/\epsilon)^{\ddim(P)}$. 
We briefly explain this argument for completeness.
Let $I_0\subseteq P$ be an optimal set of discrete centers for vanilla $k$-center of $P$,
then clearly $P\subseteq \bigcup_{x\in I_0} B(x,2D)$.
By the doubling assumption,
there exists $I_1\subseteq P$ of size $|I_1|\leq 2^{\ddim(P)} |I_0|$
such that $P\subseteq \bigcup_{x\in I_1} B(x,D)$.
Repeat this argument inductively for $\lceil\log\tfrac{2}{\epsilon}\rceil$ levels,
to cover $P$ with balls of radius at most $\epsilon D$
and obtain a set $Y\subseteq P$ of size
$|Y| =  |I_{\lceil\log\tfrac{2}{\epsilon}\rceil}| \leq (2^{\ddim(P)})^{\lceil\log\tfrac{2}{\epsilon}\rceil}|I_0|\leq k(4/\epsilon)^{\ddim(P)}$.

For suitable dimension $t=O(\epsilon^{-2}\log |Y|) = O(\epsilon^{-2}(\log(1/\epsilon) \ddim(P) + \log k))$,
by \Cref{fact:JL_lemma_one_vector} and a union bound,
with probability at least $8/9$, 
\begin{equation}\label{eq:doubling_JL_net}
    \forall y_1,y_2\in Y, \qquad
\|Gy_1-Gy_2\| \in (1\pm \epsilon)\|y_1-y_2\|.
\end{equation}
In addition, for each $y\in Y$ separately,
by \Cref{lem:IN07_doubling_radius}, with probability $1-e^{-ct}$,
\begin{equation}\label{eq:doubling_ball_no_expansion}
  \forall x\in P\cap B(y,\epsilon D),
  \qquad
  \|G(x-y)\| \leq 6\epsilon D.
\end{equation}
By a union bound, \Cref{eq:doubling_ball_no_expansion} holds for all $y\in Y$ with probability $1-e^{-ct}k(2/\epsilon)^{\ddim(P)}\geq 8/9$.
Let us assume henceforth that all these events hold simultaneously,
which occurs with probability at least $2/3$ by a union bound.

\new{
Consider a constraint $\AC$,
and let us first show that $\opt_{\AC}(G(P)) \leq (1+O(\epsilon)) \opt_{\AC}(P)$.
\new{By \Cref{eq:doubling_JL_net}, the map $G:Y\to \R^t$ is $(1+\eps)$-Lipschitz, and by \Cref{thm:kirszbraun}, it has an extension $\phi_1:\R^d\to\R^t$ that is $(1+\eps)$-Lipschitz.
Let $C^* = \{c^*_1, \dots, c^*_k\} \subset \mathbb{R}^d$ be an optimal set of $k$ centers for $P$ under assignment constraint $\AC$, 
and let $\pi^*$ be the corresponding feasible assignment.}
By using the same assignment $\pi^*$ but for points $G(P)$ and centers $\phi(C^*)$, we get
\begin{align*}
    \opt_{\AC}(G(P)) 
    &\leq \max_{p_i\in P} \|Gp_i-\phi_{1}(c^*_{\pi^*(i)})\| \\
\shortintertext{\text{and denoting by $y_p\in Y$ the nearest point to $p$ in $Y$,}} 
&\leq \max_{p_i\in P} \big[ \|Gy_{p_i}-\phi_{1}(c^*_{\pi^*(i)})\| + \|Gp_i-Gy_{p_i}\| \big] && \text{by the triangle inequality}\\ 
&\leq \max_{p_i\in P} \|Gy_{p_i}-\phi_{1}(c^*_{\pi^*(i)})\| + 6\epsilon D && \text{by \Cref{eq:doubling_ball_no_expansion}} \\
&\leq (1+\epsilon) \max_{p_i\in P} \|y_{p_i}-c^*_{\pi^*(i)}\| + 6\epsilon D && \text{by Kirszbraun's Theorem} \\
&\leq (1+\epsilon) \max_{p_i\in P} (\|y_{p_i}- p_i\|+ \|p_i - c^*_{\pi^*(i)}\|) + 6\epsilon D && \text{by the triangle inequality} \\
    &\leq (1+\epsilon) \max_{p_i\in P} \|p_i-c^*_{\pi^*(i)}\|  + 8\epsilon D  && \text{since $y_{p_i}$ is a nearest net point to $p_i$} \\
    &\leq (1+9\epsilon)\opt_{\AC}(P). 
\end{align*}

We can also show that $\opt_{\AC}(G(P))\geq (1-O(\epsilon))\opt_{\AC}(P)$. 
Indeed, the inverse map of $G:Y\to \R^t$,
i.e., mapping $Gy \mapsto y$ for all $y \in Y$, is $(1+\epsilon)$-Lipschitz by \Cref{eq:doubling_JL_net}. Thus, by the Kirszbraun Theorem (\Cref{thm:kirszbraun})
it has an extension $\phi_{-1}:\R^t\to\R^d$ that is $(1+\epsilon)$-Lipschitz. 
Let $\tilde{C}\in\R^t$ be an optimal set of $k$ centers for $G(P)$ under assignment constraint $\AC$,
and let $\tilde{\pi}$ be the corresponding feasible assignment. Then
\begin{align*}
    \opt_{\AC}(P) &\leq \max_{p_i\in P} \|(p_i-\phi_{-1}(\tilde{c}_{\tilde{\pi}(i)})) \\
    &\leq \max_{p_i\in P} \|y_{p_i}-\phi_{-1}(\tilde{c}_{\tilde{\pi}(i)})\| + \epsilon D && \text{by the triangle inequality} \\
    &\leq (1+\epsilon)\max_{p_i\in P} \|Gy_{p_i}-\tilde{c}_{\tilde{\pi}(i)}\| + \epsilon D && \text{by Kirszbraun's Theorem} \\
    &\leq (1+\epsilon)\max_{p_i\in P}(\|Gp_i-Gy_{p_i}\| + \|Gp_i-\tilde{c}_{\tilde{\pi}(i)}\|) + \epsilon D && \text{by the triangle inequality} \\
    &\leq (1+\epsilon)\opt_{\AC}(G(P)) + 13\epsilon D && \text{by \Cref{eq:doubling_ball_no_expansion} and since $\eps<1$} \\
    &\leq (1+O(\epsilon)) \opt_{\AC}(G(P)).
\end{align*}
The first bullet in \Cref{thm:doubling_diameter}
now follows by rescaling $\epsilon$.}
\new{For the second bullet, let $C=\{c_1,\ldots,c_k\}\subseteq P$ be a set of centers, $\pi$ be a feasible assignment, and suppose that $G(C)$ and $\pi$ are an $\alpha$-approximation to $k$-center with assignment constraint $\AC$ of $G(P)$.
Then, the cost of $C$ and $\pi$ with respect 
$P$ is,
\begin{align*}
    &\max_{p_i\in P} \|p_i-c_{\pi(i)}\|\\ &\leq \max_{p_i\in P} \|p_i-y_{p_i}\|+\|y_{p_i}-y_{c_{\pi(i)}}\|+\|y_{c_{\pi(i)}}-c_{\pi(i)}\| && \text{by the triangle inequality} \\
    &\leq 2\eps D + (1+O(\eps)) \max_{p_i\in P} \|Gy_{p_i}-Gy_{c_{\pi(i)}}\| && \text{by \Cref{eq:doubling_JL_net}} \\
    &\leq O(\eps D) + (1+O(\eps)) \max_{p_i\in P} \|Gp_i-Gc_{\pi(i)}\| && \text{by triangle inequality and \Cref{eq:doubling_ball_no_expansion}} \\
    &\leq (1+O(\eps))\alpha \opt_\AC(G(P)),
\end{align*}
and the proof is concluded by rescaling $\eps$.}
\end{proof}

\new{

\subsection{The JL Lemma and $k$-Center}
\label{sec:JL_preserves_constrained_kcenter}

In this section, we obtain dimension reduction for $k$-center using target dimension $O(\eps^{-2}\log n)$.
Observe that since $\ddim(P)\leq \log n$,
\Cref{thm:doubling_diameter} readily implies
a slightly weaker bound $O(\eps^{-2}\log \tfrac{1}{\eps} \cdot\log n )$.
We remove the extra $\log\tfrac{1}{\eps}$ factor by a simpler proof
that uses the same arguments as \Cref{thm:doubling_diameter},
namely, it applies the JL Lemma on $P$
and then extends the two resulting maps using the Kirszbraun Theorem,
but without using nets.
Similarly to \Cref{thm:doubling_diameter}, the result holds for
the $k$-center with an assignment constraint problem (simultaneously for all constraints).

\begin{theorem}\label{thm:JL_preserves_constrained_kcenter}
Let $P \subset \mathbb{R}^d$ be a set of $n$ points. Let $G: \mathbb{R}^d \to \mathbb{R}^t$ be a JL mapping with $t = O(\epsilon^{-2} \log n)$. Then with high probability, for all constraints $\AC:[k]^n\to \{0,1\}$,
\begin{itemize}
    \item $(1-\epsilon) \opt_{\AC}(P) \le \opt_{\AC}(G(P)) \le (1+\epsilon) \opt_{\AC}(P)$, and
    \item centers $C\subseteq P$ with assignment $\pi$ are $(\alpha(1+\epsilon))$-approximate solution for $k$-center with assignment constraint $\AC$ of $P$ whenever $G(C)$ and $\pi$ are $\alpha$-approximate solution for $k$-center with assignment constraint $\AC$ of $G(P)$.
\end{itemize}
\end{theorem}

\begin{proof}
By the JL Lemma, with high probability, the map $G$ preserves all pairwise distances in $P$ within $(1+\eps)$-factor.
Observe that $G:P\to \R^t$ is $(1+\eps)$-Lipschitz and $G^{-1}:G(P)\to \R^d$ is $(1+O(\eps))$-Lipschitz, and thus by Kirszbraun Theorem (\Cref{thm:kirszbraun}), $G$ has a $(1+\eps)$-Lipschitz extension $\phi_1:\R^d\to\R^t$, and $G^{-1}$ has a  $(1+O(\eps))$-Lipschitz extension $\phi_{-1}:\R^t\to\R^d$.
Consider a constraint $\AC$.
Let $C^* = \{c^*_1, \dots, c^*_k\} \subset \mathbb{R}^d$ and $\pi^*$ be optimal centers and feasible assignment for $P$.

For the upper bound on $\opt_{\AC}(G(P))$,
we can use the centers $\phi_1(C^*) = \{\phi_1(c^*_1), \dots, \phi_1(c^*_k)\}\subset\mathbb{R}^t$, and the optimal assignment $\pi^*$ applied now to $G(P)$. We have,
\begin{align*}
    \opt_{\AC}(G(P)) &\leq \max_{p_i \in P} \|Gp_i - \phi_1(c^*_{\pi^*(i)})\| \\
     &\leq (1+\epsilon) \max_{p_i \in P}\|p_i - c^*_{\pi^*(i)}\| && \text{$\phi_1$ is $(1+\eps)$-Lipschitz} \\
    &= (1+\epsilon) \opt_{\AC}(P).
\end{align*}
For the lower bound on $\opt_{\AC}(G(P))$, let $\tilde{C} = \{\tilde{c}_1, \dots, \tilde{c}_k\} \subset \mathbb{R}^t$ and $\tilde{\pi}$ be optimal centers and feasible assignment for $G(P)$, which we can apply also to $P$. We have,
\begin{align*}
    \opt_{\AC}(P) &\leq \max_{p_i \in P} \|p_i - \phi_{-1}(\tilde{c}_{\pi^*(i)})\| \\
     &\leq (1+O(\epsilon)) \max_{p_i \in P}\|Gp_i - \tilde{c}_{\pi^*(i)}\| && \text{$\phi_{-1}$ is $(1+O(\eps))$-Lipschitz} \\
    &= (1+O(\epsilon)) \opt_{\AC}(G(P)),
\end{align*}
concluding the first bullet of \Cref{thm:JL_preserves_constrained_kcenter}.
For the second bullet, let $C=\{c_1,\ldots,c_k\}\subseteq P$ be a set of centers, $\pi$ be a feasible assignment, and suppose that $G(C)$ and $\pi$ are an $\alpha$-approximation to $k$-center with assignment constraint $\AC$ of $G(P)$.
Then, the cost of $C$ and $\pi$ with respect to $P$ is,
\begin{align*}
    \max_{p_i\in P} \|p_i-c_{\pi(i)}\| &\leq (1+O(\eps)) \max_{p_i\in P} \|Gp_i-Gc_{\pi(i)}\| && \text{by the JL Lemma} \\
    &\leq (1+O(\eps))\alpha \opt_\AC(G(P))\\
    &\leq (1+O(\eps))\alpha \opt_\AC(P).
\end{align*}
Rescaling $\eps$ concludes the proof.

\end{proof}

}

     \new{

\section{Streaming Algorithm for $k$-Center}\label{sec:streaming_low_dim}

In this section, we prove \Cref{thm:streaming_informal}
by presenting streaming algorithms  for all $k$-center variants that we consider.
We first state a streaming algorithm for low-dimensional spaces (\Cref{thm:low_dim_stream}), 
and then use it together with our dimension reduction to prove \Cref{thm:streaming_informal}. 

The streaming algorithm for low dimension actually works by finding
a weighted set of points $S$, such that for all constraints $\AC$,
we have $\opt_{\AC}(S)=\Theta(\opt_{\AC}(P))$.
Such a set can be used, by employing any offline algorithm,
to approximate $\opt_{\AC}(P)$ and to compute an approximate solution.
(We note that one can easily strengthen the proof below to obtain $1+\eps$ approximation.)
The same algorithm already appears in \cite{BergBM23},
with a technical but crucial difference: 
Their algorithm reports grid points near $P$,
while our algorithm reports data points (i.e., $S\subseteq P$),
which is needed to lift solutions back to high dimension.
We achieve this by simply replacing their use of simple $\ell_0$-samplers
with two-level $\ell_0$-samplers. 
We also observe that their argument extends to $k$-center with assignment constraints.
The following theorem summarizes these guarantees, 
and for completeness we prove it in \Cref{sec:proof_stream_low_dim}.

\begin{theorem}\label{thm:low_dim_stream}
Let $n,d$ and $\Delta = \poly(n)$.
There exists a randomized single-pass streaming algorithm that, given a set $P \subseteq [\Delta]^d$ of $n$ points, presented as a stream of insertions and deletions of length $\poly(n)$, returns 
a weighted set $S\subseteq P$ with positive weights $\{w_p\}_{p\in S}$, representing a multiset constructed by moving every point in $P$ at most distance $O(\opt_{vanilla}(P))$.
The algorithm uses $k\cdot 2^{O(d)} \cdot \poly(\log n)$ bits of space 
and fails with probability at most $\tfrac{1}{\poly(n)}$.
\end{theorem}

If for a constraint $\AC$, we can maintain sufficient information about the multiset representation, then we obtain by \Cref{lem:move_points_capacitated_k_center} that $\opt_{\AC}(S)=\Theta(\opt_{\AC}(P))$.
For the vanilla, outliers and capacitated variants, it suffices to maintain $S$ as a weighted set, and for the fair variant, where every point $p\in P$ has color $c(p)\in [\ell]$, it suffices to maintain $\ell$ weighted sets, one set for each color.
Clearly, this is not possible for all constraints, and 
just storing an arbitrary constraint $\AC$ is infeasible (since it might require $k^n$ bits of space).
Eventually, one need a small space offline algorithm to actually approximate $\opt_{\AC}(P)$.

\paragraph{Vanilla $k$-center}
A streaming algorithm for vanilla $k$-center in high dimension now follows as a corollary.
First, discretize the matrix $G$ of Gaussians using $O(\log \Delta)=O(\log n)$ bits for each entry.
This has the effect of distorting $\|Gx-Gy\|$ additively by at most $\tfrac{\|x-y\|}{\poly(\Delta)}$,
which is negligible for $x,y\in [\Delta]^d$. 
Now, just apply on the input stream the dimension reduction of \Cref{thm:dimension_reduction_for_kcenter}, 
and feed the result into the streaming algorithm of \Cref{thm:low_dim_stream}.
The algorithm for dimension $t=O(d/\alpha^2)$ uses $\poly(k 2^{d/\alpha^2}\log n)$ bits of space,
and storing the discretized matrix $G$ requires $O(d^2\log n)$ bits.
It remains to apply on the output $S$ of \Cref{thm:low_dim_stream} a small space offline algorithm for vanilla $k$-center, like Gonzalez's algorithm~\cite{DBLP:journals/tcs/Gonzalez85}.

\begin{corollary} [Streaming Vanilla $k$-Center]
  \label{thm:streaming_vanilla}
  There is a randomized algorithm that,
  given as input numbers $\alpha,k,n,d,\Delta$
  and a set $P\subseteq [\Delta]^d$ of size at most $n$ presented as a stream of $\poly(n)$ insertions and deletions of points,
  returns an $O(\alpha)$-approximation (value and solution) to $k$-center on $P$.
  The algorithm uses 
  $\poly(k d 2^{d/\alpha^2}\log (n\Delta))$
  bits of space and fails with probability at most $1/ \poly(n)$.
\end{corollary}

\paragraph{Outliers Variant. }
Recall that $k$-center with $z$ outliers is a special case of $(k+z)$-center with the constraint that $z$ clusters have only one point, thus \Cref{thm:low_dim_stream} applies to $k$-center with $z$ outliers. Combining with an offline $O(1)$-approximation by \cite{DBLP:conf/soda/CharikarKMN01},
we obtain an algorithm for low dimensions using $(k+z) \cdot 2^{O(d)} \cdot \poly(\log n)$ bits of space. Thus, by first applying \Cref{thm:outliers_k_center} using a discretized matrix $G$ and then using the above, we obtain the following.

\begin{corollary} [Streaming Algorithm for $k$-Center with $z$ Outliers]
  \label{thm:streaming_outliers}
  There is a randomized algorithm that,
  given as input numbers $\alpha,k,z,n,d,\Delta$
  and a set $P\subseteq [\Delta]^d$ of size at most $n$ presented as a stream of $\poly(n)$ insertions and deletions of points,
  returns an $O(\alpha)$-approximation (value and solution) to $k$-center with $z$ outliers on $P$.
  The algorithm uses 
  $\poly(k (z+1)d 2^{d/\alpha^2}\log (n\Delta))$
  bits of space and fails with probability at most $1/ \poly(n)$.
\end{corollary}
Note that our space bound $(k+z) \cdot 2^{O(d)} \cdot \poly(\log n)$ for low dimensions
is worse than the $(k \cdot 2^{O(d)}+z) \cdot \poly(\log n)$ bound obtained in \cite{BergBM23}.
It nevertheless suffices for our application to high dimension,
because the target dimension $t$ in \Cref{thm:outliers_k_center} is linear in $\log (kz)$,
and thus our space bound for high dimension in \Cref{thm:streaming_outliers} 
has a $\poly(kz)$ factor anyway.

\paragraph{Capacitated and Fair Variants.}
As corollaries of \Cref{thm:constrained_k_center}, we get streaming algorithms for the capacitated and fair variants.
The capacitated variant, introduced by \cite{Bar-IlanKP93},
is a special case of $k$-center with an assignment constraint,
where the input is a set $P\subset\R^d$ of $n$ points and a maximum load $L>0$,
and the assignment constraint is that every cluster has at most $L$ points.
Thus, by applying the dimension reduction of \Cref{thm:constrained_k_center} using a discretized matrix $G$, using the streaming algorithm for low dimension of \Cref{thm:low_dim_stream} and using an offline $O(1)$-approximation by \cite{Bar-IlanKP93},

\begin{corollary}\label{thm:streaming_capacitated}
    There is a randomized algorithm that, given as input numbers $\alpha,L,k,n,d,\Delta$
  and a set $P\subseteq [\Delta]^d$ of size at most $n$ presented as a stream of $\poly(n)$ insertions and deletions of points,
  returns an $O(\alpha)$-approximation (value and a set of centers) to capacitated $k$-center on $P$.
  The algorithm uses 
  $\poly(k d 2^{d/\alpha}\log n)$
  bits of space and fails with probability at most $1/ \poly(n)$.
\end{corollary}

The fair variant is a special case of $k$-center with an assignment constraint,
and we use the specific formulation introduced in~\cite{DBLP:conf/waoa/SchmidtSS19}.\footnote{For simplicity, we present our result in the model of~\cite{DBLP:conf/waoa/SchmidtSS19},
  although our approach yields a similar result also in the more sophisticated formulation introduced in~\cite{DBLP:conf/nips/HuangJV19},
  where each point may be assigned to multiple colors.
}
The input is a set $P\subset\R^d$ of $n$ points, a coloring $c:P\to \{1,...,\ell\}$, and two numbers $0<a<1<b$.
For all $i\in[\ell]$, define $\zeta(i)= \tfrac{|\{p\in P : c(p)=i\}|}{|P|}$ as the fraction that color $i$ appears in the input $P$.
The assignment constraint is that in every cluster $C$, for every color $i\in [\ell]$, 
\[
a \cdot\zeta(i)\leq \frac{|\{p\in C : c(p)=i\}|}{|C|}\leq b \cdot\zeta(i).
\]

By \Cref{thm:low_dim_stream,thm:constrained_k_center},  and an offline algorithm by \cite{Bercea0KKRS019},
\begin{corollary}\label{thm:streaming_fairness}
    For every fixed $\ell\geq 1$, there is a randomized algorithm that, given as input $\alpha>1,a,b,k,n,d,\Delta$
  and set $P\subseteq [\Delta]^d$ of size at most $n$ presented as a stream of $\poly(n)$ insertions and deletions of points where every point is given with a color from $[\ell]$,
  returns an $O(\alpha)$-approximation (value and a set of centers) to fair $k$-center on $P$.
  The algorithm uses 
  $\ell \poly(k d 2^{d/\alpha}\log n)$
  bits of space and fails with probability at most $1/ \poly(n)$.
\end{corollary}

\subsection{Proof of \Cref{thm:low_dim_stream}}\label{sec:proof_stream_low_dim}

Our streaming implementation relies on two known algorithms:
(a) two-level $\ell_0$-sampler, an extension of the standard $\ell_0$-sampler; and
(b) an algorithm for estimating the number of non-zero frequencies (aka $\ell_0$-norm) in a dynamic stream.
We state these known results, and then prove \Cref{thm:low_dim_stream}.

\begin{lemma}[{Two-Level $\ell_0$-Sampler, \cite[Lemma 3.3]{DBLP_arxiv:CJK+22}}]
\label{lem:2level_ell0_sampler}
There is a randomized algorithm,
that given as input a matrix $M\in\R^{m \times n}$,
with $m\le n$ and integer entries bounded by $\poly(n)$,
that is presented as a stream of additive entry-wise updates,
returns an index-pair $(i,j)$ of $M$,
where $i$ is chosen uniformly at random from the non-zero rows,
and then $j$ is chosen uniformly at random from the columns that are non-zero  in that row $i$.
The algorithm uses space $\poly(\log n)$,
fails with probability at most $1 / \poly(n)$,
and can further report the corresponding row-sum $\sum_{j'} M_{i,j'}$.
\end{lemma}

\begin{lemma}[$\ell_0$-estimation \cite{KaneNW10}]
There is a randomized algorithm,
    that given as input a vector $x\in\R^n$ with integer entries bounded by $\poly(n)$,
    that is presented as a stream of additive entry-wise updates, returns a $1.1$-approximation to $\|x\|_0=|\{i\in [n] : x_i\neq 0\}|$.
    The algorithm uses space $\poly(\log n)$ and
    fails with probability at most $1 / \poly(n)$.
\end{lemma}

\begin{proof}[Proof of \Cref{thm:low_dim_stream}.]
Let $R^*$ be the optimal vanilla $k$-center value for the point set $P$. 
Our algorithm maintains a set of guesses for $R^*$. Because the points reside in $[\Delta]^d$ with $\Delta = \poly(n)$, the maximum possible distance is bounded by $\sqrt{d}\Delta$. Define a set of geometric guesses $\mathcal{R} = \{R_0, R_1, \dots, R_s\}$ such that $R_i = 2^i R_0$. Setting $R_0 = 1$ and $R_s \ge \sqrt{d}\Delta$ yields $s= O(\log(d\Delta)) = O(\log n)$ guesses.

\paragraph{Grid-Based Representatives via Two-Level Sampling.}
For every guess $R \in \mathcal{R}$, our algorithm implicitly partitions the space $\mathbb{R}^d$ into a grid $\GG_R$ where each grid cell is a $d$-dimensional cube of sidelength $\delta = \eps_0 R / \sqrt{d}$, for $\eps_0=0.1$. To maintain representatives in a dynamic stream, 
the algorithm uses two-level $\ell_0$-samplers \Cref{lem:2level_ell0_sampler}, as follows.
For every guess $R$, maintain $k\cdot 2^{O(d)}\log n$ independent two-level $\ell_0$-samplers over a $|\GG_R|\times [\Delta]^d$ matrix. Upon insertion of a point $p\in [\Delta]^d$, compute the grid cell $\GG_R(p)$ that contains $p$,
and insert $(\GG_R(p),p)$ to all the two-level samplers of guess $R$. Similarly, upon deletion of a point $p$, delete $(\GG_R(p),p)$ from the samplers.

The idea is that the samplers first pick a uniformly random non-empty grid cell, and then picks a point from that sampled cell and is given as weight the number of points in the cell. When $R\geq R^*$, maintaining multiple samplers suffices to recover the set of non-empty cells, as follows.

\begin{claim}\label{claim:packing_grid_ball}
    For every $R \ge R^*$, the number of non-empty grid cells of $\GG_R$ is at most $k\cdot2^{O(d)}$.
\end{claim}
\begin{proof}
    Recall that the entire point set $P$ can be covered by $k$ Euclidean balls of radius $R^* \le R$. We now bound the number of grid cells of $\GG_R$ intersecting a ball of radius $R$. 
We use a volume argument comparing a Euclidean $R$-ball to a cube of sidelength $\eps_0 R / \sqrt{d}$. The volume of a $d$-dimensional $R$-ball is given by:
    \[
    \mathrm{Vol}(\mathcal{B}_d(R)) = \frac{\pi^{d/2}R^d}{\Gamma(d/2 + 1)}, \]
    where $\Gamma$ is the gamma function, which extends the factorial to non-natural numbers.
    The volume of a single grid cube is:
    \[
    \mathrm{Vol}\left(\mathrm{Cube}\left(\tfrac{\eps_0 R}{\sqrt{d}}\right)\right) = \left(\frac{\eps_0 R}{\sqrt{d}}\right)^d = \frac{\eps_0^d R^d}{d^{d/2}}.
    \]
    Now, we bound the number of cells intersecting an $R$-ball. These cells are disjoint, and since the diameter of a cell is $\leq \eps_0 R$, these cells are contained in a larger ball of radius $R + \eps_0 R \le 2R$.
    Therefore, since $\Gamma(x)=\Theta(x^{x+1/2}e^{-x})$ for $x>1$ (see e.g. \cite{Batir08}), the number of these cells is at most
    \begin{align*}
        \frac{\mathrm{Vol}(\mathcal{B}_d(2R))}{\mathrm{Vol}(\mathrm{Cube}(\eps_0 R/\sqrt{d}))} 
        &= \frac{2^d \cdot \pi^{d/2} / \Gamma(d/2 + 1)}{\eps_0^d / d^{d/2}}\\
        &\leq O(1)^d\frac{d^{d/2}}{d^{d/2+3/2}}\leq O(1)^d.
    \end{align*}
    Since $P$ is covered by $k$ balls of radius $R\ge R^*$, the total number of non-empty grid cells of $\GG_R$ is at most $M = k \cdot 2^{O(d)}$,
    proving the claim. 
\end{proof}

\paragraph{Extracting the Solution.}
The algorithm maintains also for each guess 
an $\ell_0$-norm estimator that estimates the final number of non-empty cells.
At the end of the stream, use the $\ell_0$-norm estimators to identify the smallest guess $\hat{R} \in \mathcal{R}$ whose estimated number of non-empty cells is at most $1.1\cdot k \cdot 2^{O(d)}$,
and output the set of points returned by the two-level samplers at this guess,
denoted by $S \subseteq P$.
Give each $p\in S$ weight $w_p$ that is the row-sum computed by the two-level sampler (which is the number of points in $p$'s cell),
and consider each cell at most once by omitting its repetitions.

Now, we argue that with high probability, this weighted set $S$ satisfies \Cref{thm:low_dim_stream}.
For sake of analysis, fix a guess $R'$ in the range $(R^*,2R^*]$,
e.g., the smallest guess that is larger than the optimum $R^*$.
By \Cref{claim:packing_grid_ball}, $\|\GG_{R'}\|_0\leq k \cdot 2^{O(d)}$, hence with high probability, its $\ell_0$-norm estimator is at most $1.1\cdot k \cdot 2^{O(d)}$,
and in this case we must have $\hat{R}\leq R' \leq 2R^*$.
By a coupon collector's argument, with high probability, the number $k  2^{O(d)} \log n$ of samplers is sufficient to recover a point from every non-empty cell of $\GG_{\hat{R}}$. 
By our grid construction, every point $p \in P$ is in the same cell as some representative $s \in S$, meaning the distance between $p$ and $s$ is at most the diameter of the cell, which is $\sqrt{d} \cdot (\eps_0 \hat{R} / \sqrt{d}) = \eps_0 \hat{R}$. 
Moreover, the weight given to $s$ equals the number of points $p\in P$ represented by $s$. 
For all constraints $\AC$, we have that $\opt_{\AC}(P)\geq \opt_{vanilla}(P)=R^*$.
By the ``movement'' argument of \Cref{lem:move_points_capacitated_k_center}, we obtain $\opt_{\AC}(P)\in (1\pm 2\eps_0)\opt_{\AC}(S)$, as desired.

\paragraph{Space.}
Each sampler and $\ell_0$-norm estimator require $\poly(\log n)$ space. Therefore, the total space used per guess is $k \cdot 2^{O(d)} \cdot \poly(\log n)$, yielding an overall space complexity of $k \cdot 2^{O(d)} \cdot \poly(\log n)$ across all $O(\log n)$ guesses.
This completes the proof of \Cref{thm:low_dim_stream}.
\end{proof}

}

     \paragraph*{Acknowledgments. }
    We thank Sepideh Mahabadi for pointing out earlier work on robust coresets; and Chris Schwiegelshohn and Sandeep Silwal for pointing out to use the Kirszbraun Theorem in the proofs of \Cref{thm:doubling_diameter,thm:JL_preserves_constrained_kcenter}.  
    {\small 
    \bibliographystyle{alphaurl}
    \bibliography{references.bib}
    } \appendix

\end{document}